\documentclass[journal]{IEEEtran}
\usepackage[colorlinks=false]{hyperref}
\usepackage{pgf}
\usepackage{multimedia}
\usepackage{amsmath,amssymb,amsthm,mathrsfs,epsfig,url,bbm,algorithmic}
\usepackage{fancybox}
\usepackage{cancel}
\usepackage{graphicx}
\usepackage[cmtip,arrow]{xy}
\usepackage{pb-diagram,pb-xy}
\usepackage{wrapfig}
\usepackage{cite}
\usepackage{mathtools}
\usepackage{caption}
\usepackage{subcaption}

\usepackage{color}

\newtheorem{lem}{Lemma}
\newtheorem{thm}{Theorem}
\newtheorem{prop}{Proposition}
\newtheorem{cor}{Corollary}

\newtheorem{rmk}{Remark}
\newtheorem{model}{Model}
\newcommand{\R}{\mathbb{R}}

\newcommand{\Ex}{\mathbb{E}}
\newcommand{\Var}{\text{Var}}
\DeclarePairedDelimiterX{\norm}[1]{\lVert}{\rVert}{#1}

\newcommand{\Lb}{\mathbf{L}}

\DeclareMathOperator*{\argmin}{arg\,min}
\newcommand{\varg}{\mathring{g}}
\newcommand{\varrootg}{\mathring{p}}
\newcommand{\vareta}{\mathring{\eta}}
\newcommand{\snr}{\text{SNR}}
\newcommand{\Cb}{\mathbf{C}}

\ifCLASSINFOpdf
\else
\fi
\begin{document}
%
\title{Unbiasing Procedures for Scale-invariant Multi-reference Alignment}
%
%
%

\author{Matthew~Hirn, 
        Anna~Little 
        
\thanks{M. Hirn is with the Department of Computational Mathematics, Science and Engineering, the Department of Mathematics, and the Center for Quantum Computing, Science and Engineering, Michigan State University, East Lansing,
MI, 48824 USA, e-mail: mhirn@msu.edu.}
\thanks{A. Little is with the Department of Mathematics and the Utah Center For Data Science, University of Utah, Salt Lake City, UT, 84112 USA, e-mail: little@math.utah.edu.}
}

\maketitle

\begin{abstract}
This article discusses a generalization of the 1-dimensional multi-reference alignment problem. The goal is to recover a hidden signal from many noisy observations, where each noisy observation includes a random translation and random dilation of the hidden signal, as well as high additive noise. We propose a method that recovers the power spectrum of the hidden signal by applying a data-driven, nonlinear unbiasing procedure, and thus the hidden signal is obtained up to an unknown phase. An unbiased estimator of the power spectrum is defined, whose error depends on the sample size and noise levels, and we precisely quantify the convergence rate of the proposed estimator. The unbiasing procedure relies on knowledge of the dilation distribution, and we implement an optimization procedure to learn the dilation variance when this parameter is unknown. Our theoretical work is supported by extensive numerical experiments on a wide range of signals.  
\end{abstract}


%
\IEEEpeerreviewmaketitle

\section{Introduction}

In classic multi-reference alignment (MRA), one attempts to recover a hidden signal $f : \R \rightarrow \R$ from many noisy observations, where each noisy observation has been randomly
translated and corrupted by additive noise, as described in the following model.
\begin{model}[Classic MRA]
	\label{model:classicMRA}
	The \textit{classic MRA data model} consists of $M$ independent observations of a compactly supported, real-valued signal $f \in \Lb^2(\R)$:
	\begin{equation} 
		y_j(x) = f(x-t_j) + \varepsilon_j(x) \, , \quad 1 \leq j \leq M \, ,
	\end{equation}	
	where:
	\begin{itemize}
		\item[(i)] $\text{supp}(y_j)\subseteq [-\frac{1}{2},\frac{1}{2}]$ for $1 \leq j \leq M$.
		\item[(ii)] $\{t_j\}_{j=1}^M$ are independent samples of a random variable $t\in \R$.
		\item[(iii)] $\{\varepsilon_j(x)\}_{j=1}^M$ are independent white noise processes on $[-\frac{1}{2},\frac{1}{2}]$ with variance $\sigma^2$.
	\end{itemize}
\end{model}
This toy model is a first step towards more realistic models arising in cryo-electron microscopy, and is relevant in many other applications including structural biology \cite{theobald2012optimal, diamond1992multiple, scheres2005maximum, sadler1992shift, park2011stochastic, park2014assembly}; radar \cite{zwart2003fast, gil2005using}; single cell genomic sequencing \cite{leggett2015nanook}; image registration \cite{foroosh2002extension, brown1992survey, robinson2007optimal}; and signal processing \cite{zwart2003fast}. Some methods solve Model \ref{model:classicMRA} via \textit{synchronization} \cite{singer2011angular, boumal2016nonconvex, perry2018message, chen2018projected, bandeira2017tightness, zhong2018near, bandeira2020non, bandeira2014multireference, chen2014near, bandeira2016low}, i.e. the translation factors $\{t_j\}_{j=1}^M$ are explicitly recovered and the signals aligned. Synchronization approaches will fail in the high noise regime when the signal-to-noise ratio (SNR) is low, but the hidden signal can still be recovered by methods which avoid alignment; these include the \textit{method of moments} \cite{hansen1982large, kam1980reconstruction, sharon2019method}, which contain the \textit{method of invariants} \cite{bendory2017bispectrum, OptConvRates_MRA, collis1998higher} as a special case, and \textit{expectation-maximization type algorithms} \cite{dempster1977maximum, abbe2018multireference}. The method of invariants leverages translation invariant Fourier features such as the power spectrum and bispectrum,  as they are especially useful for solving Model \ref{model:classicMRA}. Recall the Fourier transform of a signal $f\in\Lb^1 (\R)$ is defined as
\begin{equation*}
	\widehat{f} (\omega) = \int  f(x) e^{- i x \omega} \, dx \, ,
\end{equation*}
and its power spectrum is then defined by $(Pf)(\omega) = |\widehat{f}(\omega)|^2$.

In this article we analyze the following generalization of classic MRA, where signals are also corrupted by a random scale change (i.e. dilation) in addition to random translation and additive noise. See Figure \ref{fig:Model}.

\begin{figure}[tb]
	\centering
	\includegraphics[width=.49\textwidth]{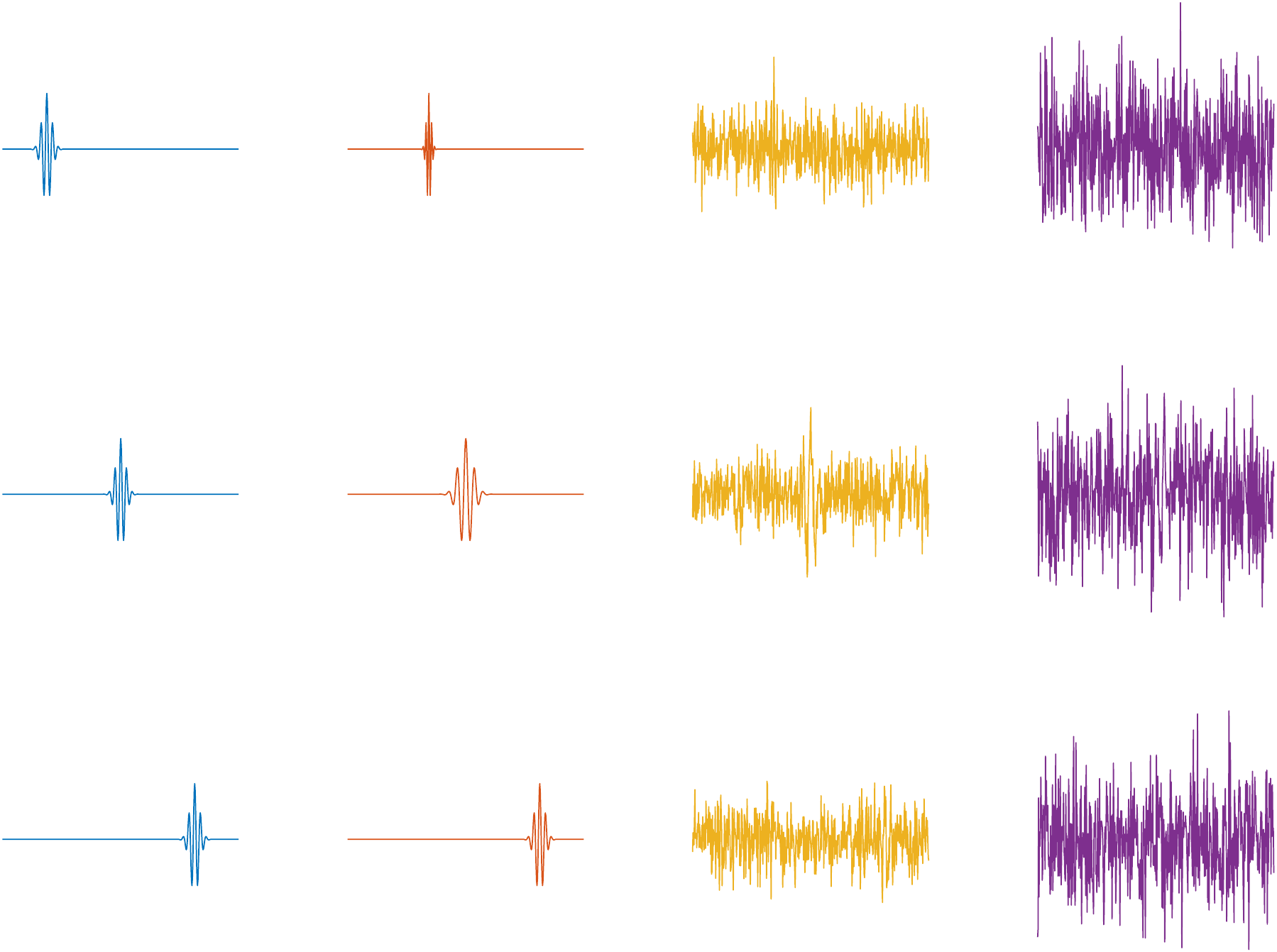}
	\caption{Model illustration: a hidden signal is randomly translated (first column), randomly dilated (second column), and then corrupted by additive noise (second and third columns). Column 2 shows corruption with $\sigma^2=\frac{1}{2}$ and Column 3 with $\sigma^2=2$; the purple curves illustrate the noise level considered in the simulations reported in Section \ref{sec:simu_results}.}
	\label{fig:Model}
\end{figure}

\begin{model}[Noisy dilation MRA data model] 
	\label{model:genMRA}
	The \textit{noisy dilation MRA data model} consists of $M$ independent observations of a compactly supported, real-valued signal $f \in \Lb^2(\R)$:
	\begin{equation} \label{eqn: translation dilation noise model}
		y_j (x) = f\left( (1 - \tau_j)^{-1}(x - t_j)\right) + \varepsilon_j(x) \, , \quad 1 \leq j \leq M \, .
	\end{equation}	
	In addition, we assume:
	\begin{itemize}
		\item[(i)] $\text{supp}(y_j)\subseteq [-\frac{1}{2},\frac{1}{2}]$ for $1 \leq j \leq M$.
		\item[(ii)] $\{t_j\}_{j=1}^M$ are independent samples of a random variable $t\in \R$.
		\item[(iii)] $\{\tau_j\}_{j=1}^M$ are independent samples from a uniformly distributed random variable $\tau$ satisfying:
		\begin{equation*}
			\tau \in \R\quad,\quad \Ex(\tau) = 0 \quad,\quad \Var(\tau) = \eta^2\leq 1/12.
		\end{equation*} 
		\item[(iv)] $\{\varepsilon_j(x)\}_{j=1}^M$ are independent white noise processes on $[-\frac{1}{2},\frac{1}{2}]$ with variance $\sigma^2$.
	\end{itemize}
\end{model}	
Model \ref{model:genMRA} is a first step towards studying more general diffeomorphisms $f(\tau(x))$, since it considers the case when $\tau(x)$ is an affine function. This is relevant to molecular imaging applications since the flexible regions of macro-molecular structures create diffeomorphisms of the underlying shape \cite{palamini2016identifying}. Dilations are also highly relevant in imaging applications, for example \cite{chandran1992position, capodiferro1987correlation, tsatsanis1990translation, hotta2001scale, martinec2007robust, robinson2007optimal}. Note we consider $\Lb^\infty (\R)$ normalized dilations in Model \ref{model:genMRA} since this is natural for images; however the method is easily modified to accommodate other normalizations such as $\Lb^1 (\R)$. The latter is useful in the statistical context, where one may observe samples from a family of distributions which are shifts and rescalings of an underlying distribution.

Solving Model \ref{model:genMRA} is highly challenging. Dilations cause instabilities in the high frequencies of a signal, where even a small dilation can lead to a large perturbation of the frequency values. Ideally, one would like to compute a representation which (1) is both translation and dilation invariant, (2) allows for the additive noise to be removed by averaging, and (3) is invertible with a numerically stable algorithm. However there is a tension between achieving (1) and (3), since the more invariants which are built into the representation, the harder it will be to invert the representation and obtain the underlying signal. In this article we propose the following compromise: we do not define a dilation invariant representation, but propose a method for dilation unbiasing which can be achieved with a numerically stable algorithm; we learn the power spectrum of the hidden signal instead of the hidden signal itself, thus reducing to a phase retrieval problem.  

Model \ref{model:genMRA} was considered in \cite{hirn2020wavelet} for $\Lb^1 (\R)$ normalized dilations. The authors define wavelet-based, translation invariant features and unbias for dilations by utilizing the first few moments of the dilation distribution. The method has two main short-comings: although it can reduce the bias due to dilations, it cannot remove it entirely, i.e. the method of \cite{hirn2020wavelet} does not define an \textit{unbiased estimator} of the true features. In addition, inverting the wavelet-based features to recover the power spectrum of the hidden signal is numerically unstable, as it is driven by the condition number of a low rank matrix. This article proposes a method which overcomes both of these challenges: by working directly on the power spectrum, we avoid a numerically unstable inversion, and we develop a new unbiasing procedure which yields an \textit{unbiased estimator} of the power spectrum of the hidden signal; we refer to this unbiasing procedure as \textit{inversion unbiasing}. To achieve this we assume explicit knowledge of the dilation distribution instead of knowledge of the first few moments. 
To illustrate inversion unbiasing, it is helpful to define the following model in which signals are randomly translated and dilated, but not corrupted by additive noise.
\begin{model}[Dilation MRA data model] 
	\label{model:dilMRA}
	The \textit{dilation MRA data model} consists of $M$ independent observations of a compactly supported, real-valued signal $f \in \Lb^2(\R)$:
	\begin{equation} \label{eqn: translation dilation model}
		y_j (x) = f\left( (1 - \tau_j)^{-1}(x - t_j)\right) \, , \quad 1 \leq j \leq M \, .
	\end{equation}	
	In addition, we assume (i)-(iii) of Model \ref{model:genMRA}.
\end{model}	
Since Model \ref{model:dilMRA} lacks additive noise, it can in fact be trivially solved by first estimating $\|f\|_2$, and then dilating any observed signal to have the right norm (for further details see \cite{hirn2020wavelet}). We use Model \ref{model:dilMRA} to build a theory to solve Model \ref{model:genMRA}, but note it is not of independent interest.

\begin{rmk}
	The box size in Models \ref{model:classicMRA}--\ref{model:dilMRA} is arbitrary; more generally, the signals may be supported on any finite interval $[-\frac{N}{2}, \frac{N}{2}]$. All results still hold with $\sigma\sqrt{N}$ replacing $\sigma$. 
\end{rmk}	

The remainder of the article is organized as follows. Section \ref{sec:framework} motivates inversion unbiasing by first considering the infinite sample size case. Section \ref{sec:main} presents our main results for solving Models \ref{model:genMRA} and \ref{model:dilMRA} in the finite sample regime. Section \ref{sec:opt} discusses how inversion unbiasing is implemented via an optimization algorithm. Section \ref{sec:simu_results} reports simulation results testing the performance of inversion unbiasing. Section \ref{sec:conc} concludes the article and summarizes future research directions.

\subsection{Notation}

Let $f_j(x) = f\left((1-\tau_j)^{-1}(x-t_j)\right)$ denote the $j^{\text{th}}$ signal which is dilated by $1-\tau_j$. We note that $\widehat{f}_j(\omega) = (1-\tau_j)\widehat{f}((1-\tau_j)\omega)$, so that 
\begin{align*}
	(Pf_j)(\omega) &= (1-\tau_j)^2 (Pf)((1-\tau_j)\omega)\, .
\end{align*}
We let $g = Pf$, and for Models \ref{model:genMRA} and \ref{model:dilMRA} we define
\begin{align}
	\label{equ:g_eta}
	g_\eta(\omega) &:= 
	\Ex_\tau \left[(Pf_j)(\omega)\right] \, .
\end{align}
Note for Model \ref{model:genMRA}, it is easy to show that
\begin{align*}
	g_\eta(\omega) &= \Ex_{\tau,\epsilon}[(Py_j)(\omega) - \sigma^2] \, ;
\end{align*}
see for example Proposition 3.1 in \cite{hirn2020wavelet}. We also define the following constants which depend on $\eta$:
\begin{align}
	\label{equ:Constants}
	C_0 &= \frac{(1-\sqrt{3}\eta)}{(1+\sqrt{3}\eta)}\, , \ C_1 = 2\sqrt{3}\eta \, , \ C_2 = \frac{1}{1+\sqrt{3}\eta} \, ,
\end{align}
and we let $(L_C g)(\omega) = C^3 g(C\omega)$ be a dilation operator. We use $a^*$ to denote the complex conjugate of $a$. Finally, $a=O(b)$ denotes $a\leq Cb$ for an absolute constant $C$.

\section{Infinite sample estimate}
\label{sec:framework}

To motivate our finite sample procedure, we first consider how to define an unbiased estimator in the infinite sample limit. We can recover $Pf$ from $g_\eta$, as stated in the following Proposition. 

\begin{prop}
	\label{prop:InfiniteSampleRecoveryPS}
	Assume $Pf \in \mathbf{C}^0 (\R)$ and $g_\eta$ as defined in \eqref{equ:g_eta}. Then for $\omega \ne 0$:
	\begin{align*}
		(Pf)(\omega) &= (I-L_{C_0})^{-1}C_1L_{C_2}(3g_\eta(\omega)+\omega g_\eta'(\omega)) \, ,
	\end{align*}
	where $C_0, C_1, C_2$ are as defined in (\ref{equ:Constants}).
\end{prop} 
\begin{proof}[Proof of Proposition \ref{prop:InfiniteSampleRecoveryPS}]
	Since $\tau$ has a uniform distribution with variance $\eta^2$, the pdf of $\tau$ has form $p_\tau= \frac{1}{2\sqrt{3}\eta} \mathbbm{1}[-\sqrt{3}\eta, \sqrt{3}\eta]$. Thus:
	\begin{align*}
		g_\eta(\omega) &:= \Ex_\tau[(1-\tau)^2g((1-\tau)\omega)] \\
		&= \int (1-\tau)^2g((1-\tau)\omega) p_\tau(\omega)\ d\tau \\
		&= \frac{1}{2\sqrt{3}\eta}\int_{-\sqrt{3}\eta}^{\sqrt{3}\eta} (1-\tau)^2g((1-\tau)\omega)\ d\tau \\
		&=\frac{1}{2\sqrt{3}\eta}\int_{(1-\sqrt{3}\eta)\omega}^{(1+\sqrt{3}\eta)\omega} \frac{\widetilde{\tau}^2}{\omega^2} g(\widetilde{\tau})\frac{1}{\omega}\ d\widetilde{\tau}\, ,
	\end{align*}
	where we have applied the change of variable $\widetilde{\tau}=(1-\tau)\omega$,
	$d\tau = -\frac{1}{\omega}\, d\widetilde{\tau}$. Letting $h(x)=x^2g(x)$ and $H(x)$ an antiderivative of $h$, by the Fundamental Theorem of Calculus we thus obtain:
	\begin{align*}
		2\sqrt{3}\eta \omega^3g_\eta(\omega) &= \int_{(1-\sqrt{3}\eta)\omega}^{(1+\sqrt{3}\eta)\omega} \widetilde{\tau}^2 g(\widetilde{\tau})\ d\widetilde{\tau} \\
		&= \int_{(1-\sqrt{3}\eta)\omega}^{(1+\sqrt{3}\eta)\omega} h(\widetilde{\tau})\ d\widetilde{\tau} \\
		&= H((1+\sqrt{3}\eta)\omega)-H((1-\sqrt{3}\eta)\omega) \, .
	\end{align*}
	Differentiating with respect to $\omega$ yields:
	\begin{align*}
		&2\sqrt{3}\eta \left(3\omega^2g_\eta(\omega)+\omega^3g_\eta'(\omega)\right) \\
		&\ = (1+\sqrt{3}\eta)h((1+\sqrt{3}\eta)\omega) - (1-\sqrt{3}\eta)h((1-\sqrt{3}\eta)\omega) \, ,
	\end{align*}
	and dividing by $\omega^2$ gives:
	\begin{align*}
		&2\sqrt{3}\eta \left(3g_\eta(\omega)+\omega g_\eta'(\omega)\right) \\
		&\ = (1+\sqrt{3}\eta)^3g((1+\sqrt{3}\eta)\omega) - (1-\sqrt{3}\eta)^3g((1-\sqrt{3}\eta)\omega) \, .
	\end{align*}
	Applying the dilation operator $L_{C_2}$ then gives:
	\begin{align*}
		&C_1L_{C_2}(3g_\eta+\omega g_\eta'(\omega)) \\
		&\quad =g\left(\omega\right)- \left(\frac{1-\sqrt{3}\eta}{1+\sqrt{3}\eta}\right)^3 g\left(\left(\frac{1-\sqrt{3}\eta}{1+\sqrt{3}\eta}\right)\omega\right) \\
		&\quad = (I-L_{C_0})g \, .
	\end{align*}
	Since $C_0 <1$, the series $I + L_{C_0} +L_{C_0}^2+L_{C_0}^3+\ldots$ converges, and $I-L_{C_0}$ is invertible. We thus obtain
	\begin{align*}
		g &= (I-L_{C_0})^{-1}C_1L_{C_2}(3g_\eta+\omega g_\eta'(\omega)) \, ,
	\end{align*}
	which proves the proposition.
\end{proof}


\section{Finite sample estimates}
\label{sec:main}

Since we are only given a finite sample, we do not have access to $g_\eta$, but for large $M$, $g_\eta$ is well approximated by:
\begin{align}
	\label{equ:estimate_g_eta}
	\widetilde{g}_\eta(\omega) &:= \frac{1}{M} \sum_{j=1}^M (Pf_j)(\omega) \, .
\end{align}
For dilation MRA, $\widetilde{g}_\eta$ can be computed exactly, and we describe the resulting estimator in Section \ref{sec:DilationMRA}. For noisy dilation MRA, $\widetilde{g}_\eta$ cannot be computed exactly due to the additive noise, but an unbiased estimator can still be defined as described in Section \ref{sec:NoisyDilationMRA}.

\subsection{Results for Dilation MRA}
\label{sec:DilationMRA}

Motivated by Propositions \ref{prop:InfiniteSampleRecoveryPS}, we define the following estimator for dilation MRA: 
\begin{align}
	\label{equ:FiniteSamplePSest_dilMRA}
	(\widetilde{Pf})(\omega) := (I-L_{C_0})^{-1}C_1L_{C_2}(3\widetilde{g}_\eta(\omega)+\omega\widetilde{g}_\eta'(\omega))\, ,
\end{align}
where $\widetilde{g}_\eta$ is as defined in \eqref{equ:estimate_g_eta}.
We note that in practice one does not have a closed form formula for applying $(I-L_{C_0})^{-1}$, but $(\widetilde{Pf})(\omega)$ can be obtained by solving the following convex optimization problem:
\begin{align*}
	\argmin_{\varg}\ \norm{(I-L_{C_0})\varg -C_1L_{C_2}(3\widetilde{g}_\eta(\omega)+\omega \widetilde{g}_\eta'(\omega)) }_2^2 \, .
\end{align*}
We describe this optimization procedure in detail in Section \ref{sec:opt}, but first we analyze the statistical properties of the estimator $(\widetilde{Pf})(\omega)$. The key quantity we bound is the mean squared error (MSE) $ \Ex \left[\norm{Pf-\widetilde{Pf}}_{2}^2\right]$.
The following lemma establishes that when $\widetilde{g}_\eta, \widetilde{g}_\eta'$ are good approximations of $g_\eta, g_\eta'$,  $\widetilde{Pf}$ is a good approximation of $Pf$, so we can reduce the problem to controlling $\widetilde{g}_\eta, \widetilde{g}_\eta'$. 
\begin{lem}
	\label{lem:SimplifyError}
	Assume Model \ref{model:dilMRA}, $Pf\in \mathbf{C}^1 (\R)$, and the estimator $(\widetilde{Pf})(\omega)$ defined in \eqref{equ:FiniteSamplePSest_dilMRA}. Then:
	\begin{align*}
		\norm{Pf-\widetilde{Pf}}_{2}^2 &\lesssim \norm{g_\eta-\widetilde{g}_\eta}_{2}^2 + \norm{\omega( g_\eta'(\omega)-\widetilde{g}_\eta'(\omega))}_{2}^2 \, .
	\end{align*}
\end{lem}
\begin{proof}
	From Proposition \ref{prop:InfiniteSampleRecoveryPS} and (\ref{equ:FiniteSamplePSest_dilMRA})
	\begin{align*}
		Pf &- \widetilde{Pf} = \\ &(I-L_{C_0})^{-1}C_1L_{C_2}\left[3(g_\eta-\widetilde{g}_\eta)+\omega( g_\eta'(\omega)-\widetilde{g}_\eta'(\omega))\right] \, .
	\end{align*}
	Letting $\norm{\cdot}$ denote the spectral norm, we thus obtain:
	\begin{align*}
		\norm{Pf-\widetilde{Pf}}_2^2 &\leq C_1^2 \norm{(I-L_{C_0})^{-1}}^2  \norm{L_{C_2}}^2 \\ 
		&\qquad \times \norm{3(g_\eta-\widetilde{g}_\eta)+\omega( g_\eta'(\omega)-\widetilde{g}_\eta'(\omega))}_2^2 \\
		&\leq 2C_1^2 \norm{(I-L_{C_0})^{-1}}^2  \norm{L_{C_2}}^2 \\
		&\qquad \times \left( 9\norm{g_\eta-\widetilde{g}_\eta}_2^2 + \norm{\omega( g_\eta'(\omega)-\widetilde{g}_\eta'(\omega))}_2^2 \right)\, .
	\end{align*}
	We first observe that $\norm{L_C^i}=C^\frac{5i}{2}$ since 
	\begin{align*}
		\norm{L_C^i g}_2^2 &= \int (C^{3i} g(C^i\omega))^2\ d\omega \\
		&= \int C^{6i} g(\tilde{\omega})^2\ \frac{d\tilde{\omega}}{C^i} \quad\text{for}\quad \tilde{\omega} = C^i \omega\\
		&= C^{5i} \norm{g}_2^2\, .
	\end{align*}
	Thus
	\begin{align*}
		\norm{(I-L_{C_0})^{-1}} &= \norm[\big]{ \sum_{i=0}^{\infty} L_{C_0}^i } \leq 
		\sum_{i=0}^{\infty} C_0^{\frac{5i}{2}} = \frac{1}{1-C_0^{\frac{5}{2}}} = O(\eta^{-1}) \\
		\norm{L_{C_2}} &= C_2^{\frac{5}{2}} = O(1) \\
		C_1 &= O(\eta)
	\end{align*}
	so that 
	\begin{align*}
		2C_1^2 \norm{(I-L_{C_0})^{-1}}^2  \norm{L_{C_2}}^2 &= O(1)O(\eta^{2})O(\eta^{-2}) = O(1)
	\end{align*}
	and we obtain
	\begin{align*}
		\norm{Pf-\widetilde{Pf}}_2^2 &\lesssim  \norm{g_\eta-\widetilde{g}_\eta}_2^2 + \norm{\omega( g_\eta'(\omega)-\widetilde{g}_\eta'(\omega))}_2^2 \, .
	\end{align*}
\end{proof}

Lemma \ref{lem:SimplifyError} thus establishes that to bound $\Ex \left[\norm{Pf-\widetilde{Pf}}_{2}^2\right]$, it is sufficient to bound $\Ex \left[ \norm{g_\eta-\widetilde{g}_\eta}_{2}^2 \right]$
and $\Ex \left[ \norm{\omega( g_\eta'(\omega)-\widetilde{g}_\eta'(\omega))}_{2}^2\right]$. 
Utilizing Lemma \ref{lem:SimplifyError} yields the following Theorem, which bounds the MSE of \eqref{equ:FiniteSamplePSest_dilMRA} for dilation MRA. To control higher order terms we define:
\begin{align*}
	(\overline{Pf})^{k}(\omega) &:= \max_{\xi \in [\omega/2, 2\omega]} |(Pf)^{k}(\xi)| \, .
\end{align*}
In general for well behaved functions $(\overline{g})^{k}$ and $g^{k}$ have the same decay rate; for example, if $g^k$ is monotonic, $(\overline{g})^{k}(\omega) = g^{k}(2\omega)$.

\begin{thm}
	\label{thm:DilMRA_PS}
	Assume Model \ref{model:dilMRA}, the estimator $(\widetilde{Pf})(\omega)$ defined in \eqref{equ:FiniteSamplePSest_dilMRA}, $Pf \in \mathbf{C}^3 (\R)$, and that $\omega^k(\overline{Pf})^{(k)}(\omega) \in \Lb^2 (\R)$ for $k=2,3$. Then:
	\begin{align*}
		\Ex \left[ \norm{Pf-\widetilde{Pf}}_{2}^2 \right] \lesssim \frac{\eta^2}{M}\big(&\norm{(Pf)(\omega)}_2^2 + \norm{\omega (Pf)'(\omega)}_2^2 \\
		&+\norm{\omega^2 (Pf)''(\omega)}_2^2 \big) + r\, ,
	\end{align*}
	where $r$ is a higher-order term satisfying
	\begin{align*}
		r &\leq \frac{\eta^4}{M}\left(\norm{\omega^2(\overline{Pf})^{''}(\omega)}_2^2+ \norm{\omega^3(\overline{Pf})^{'''}(\omega)}_2^2 \right)\, .
	\end{align*}
\end{thm}
\begin{proof}
	By Lemma \ref{lem:SimplifyError}, it is sufficient to bound $\Ex \left[\norm[\big]{g_\eta-\widetilde{g}_\eta}_{2}^2\right]$ and $\Ex\left[\norm[\big]{\omega( g_\eta'(\omega)-\widetilde{g}_\eta'(\omega))}_{2}^2\right]$.
	Since $\widetilde{g}_\eta(\omega) = \frac{1}{M}\sum_{j=1}^M Pf_j(\omega)$, we have
	\begin{align*}
		(\widetilde{g}_\eta(\omega) - g_\eta(\omega))^2 &\leq \left( \frac{1}{M}\sum_{j=1}^M (Pf_j)(\omega) - g_\eta(\omega)\right)^2 \, .
	\end{align*}
	Let $X_j = (Pf_j)(\omega) -g_\eta(\omega)=(Pf_j)(\omega)  - \Ex\left[ (Pf_j)(\omega) \right]$. 	Thus because $\frac{1}{M}\sum_{j=1}^M X_j$ is a centered random variable, we have
	\begin{align}
		\label{equ:termI}
		\Ex\left[ \left(\frac{1}{M}\sum_{j=1}^M X_j\right)^2 \right] &= \text{var}\left[\frac{1}{M}\sum_{j=1}^M X_j\right] 
		=\frac{\text{var}(X_j)}{M}\, .
	\end{align}
	Note that we can write:
	\begin{align*} 
		X_j &= (Pf_j)(\omega)  -  (Pf)(\omega)+ (Pf)(\omega) - \Ex\left[(Pf_j)(\omega) \right] \\
		X_j^2 &\leq 2\left((Pf_j)(\omega)  -(Pf)(\omega)\right)^2 \\
		&\qquad + 2\left((Pf)(\omega) - \Ex\left[(Pf_j)(\omega) \right]\right)^2
	\end{align*}
	Since it is easy to check that 
	\begin{align*}
		\Ex&\left[ \left((Pf)(\omega) - \Ex\left[(Pf_j)(\omega)\right]\right)^2\right] \\
		&\qquad \leq \Ex\left[ \left((Pf_j)(\omega) - (Pf)(\omega)\right)^2\right] \, ,
	\end{align*} 
	we obtain
	\begin{align*}
		\Ex\left[ X_j^2 \right] &\leq 4 \Ex \left[\left((Pf_j)(\omega) - (Pf)(\omega)\right)^2\right] \, .
	\end{align*}
	Taylor expanding $(Pf)((1-\tau_j)\omega)$ gives:
	\begin{align*}
		(Pf)((1-\tau_j)\omega) &= (Pf)(\omega) + (Pf)'(\omega) \cdot \omega\tau_j \\
		&\quad \pm \frac{1}{2}(\overline{Pf})^{''}(\omega) \cdot \omega^2\tau_j^2 \, .
	\end{align*}
	Multiplying by $(1-\tau_j)^2$ and rearranging:
	\begin{align*}
		&(1-\tau_j)^2(Pf)((1-\tau_j)\omega) - (Pf)(\omega) \\
		&\quad = (-2\tau_j+\tau_j^2)(Pf)(\omega) + (1-\tau_j)^2 (Pf)'(\omega) \cdot \omega\tau_j \\
		&\qquad \pm \frac{(1-\tau_j)^2}{2}(\overline{Pf})^{''}(\omega) \cdot \omega^2\tau_j^2 \, .
	\end{align*}
	Utilizing $a+b-c \leq d \leq a+b+c \implies d^2 \lesssim a^2+b^2+c^2$, we square and take expectation to obtain
	\begin{align*}
		&\Ex\left[\left((Pf_j)(\omega) - (Pf)(\omega)\right)^2\right]  \\ &\lesssim \left[(Pf)(\omega)\right]^2\eta^2
		+ \left[\omega (Pf)'(\omega)\right]^2 \eta^2 + \left[\omega^2(\overline{Pf})^{''}(\omega)\right]^2 \eta^4  \, .
	\end{align*}
	Thus
	\begin{align*}
		\text{var}[X_j]=\Ex[X_j^2] &\lesssim \left( \left[(Pf)(\omega)\right]^2 + \left[\omega (Pf)'(\omega)\right]^2\right) \eta^2 \\
		&\qquad+ \left[\omega^2(\overline{Pf})^{''}(\omega)\right]^2 \eta^4 \, .
	\end{align*}
	Utilizing (\ref{equ:termI}), we obtain
	\begin{align*}
		&\Ex \left[(\widetilde{g}_\eta(\omega) - g_\eta(\omega))^2\right] \\
		&\quad \lesssim \frac{\eta^2}{M} \left(\left[(Pf)(\omega)\right]^2+\left[\omega (Pf)'(\omega)\right]^2 + \left[\omega^2(\overline{Pf})^{''}(\omega)\right]^2 \eta^2\right)\, 
	\end{align*}
	so that
	\begin{align*}
		&\Ex \left[ \norm{g_\eta-\widetilde{g}_\eta}_{2}^2 \right] = \int \Ex \left[(\widetilde{g}_\eta(\omega) - g_\eta(\omega))^2\right]\ d\omega \\
		&\ \lesssim \frac{\eta^2}{M}\left(\norm{ (Pf)(\omega)}_2^2 + \norm{ \omega (Pf)'(\omega)}_2^2 + \norm{\omega^2 (\overline{Pf})^{''}(\omega)}_2^2\, \eta^2\right) \, .
	\end{align*}
	We now  bound $\Ex\left[\norm[\big]{\omega( g_\eta'(\omega)-\widetilde{g}_\eta'(\omega))}_{2}^2\right]$. Letting $g_j = Pf_j$, we have
	\begin{align*}
		\omega\widetilde{g}'_\eta(\omega) -\omega g'_\eta(\omega) &= \frac{1}{M}\sum_{j=1}^M \omega g'_j(\omega) - \omega g'_\eta(\omega) = \frac{1}{M}\sum_{j=1}^M Z_j
	\end{align*}
	where
	\begin{align*}
		Z_j &=\omega g'_j(\omega) - \omega g'_\eta(\omega) \ .
	\end{align*}
	We note $\Ex[Z_j]=0$, and a similar argument as the one applied to $X_j$ gives
	\begin{align*}
		Z_j^2 &\leq 2\left(\omega g'_j(\omega) - \omega g'(\omega) \right)^2 + 2\left(\omega g'(\omega) - \omega g'_\eta(\omega)\right)^2 \\
		\Ex\left[Z_j^2\right] &\leq 4\Ex\left[\left(\omega g'_j(\omega) - \omega g'(\omega) \right)^2\right] \, .
	\end{align*}
	Taylor expanding $(Pf)'((1-\tau_j)\omega)$ gives
	\begin{align*}
	(Pf)'((1-\tau_j)\omega) &= (Pf)'(\omega) + (Pf)''(\omega) \cdot \omega\tau_j \\
	&\quad \pm \frac{1}{2}(\overline{Pf})^{'''}(\omega) \cdot \omega^2\tau_j^2 \, .
	\end{align*}
	Since $\omega g_j'(\omega) = \omega (Pf_j)'(\omega)=(1-\tau_j)^3\omega (Pf)'((1-\tau_j)\omega)$, we multiply by $(1-\tau_j)^3\omega$ to obtain:
	\begin{align*}
		\omega (Pf_j)'(\omega) &= (1-\tau_j)^3\omega(Pf)'(\omega) +  \tau_j(1-\tau_j)^3 \omega^2 (Pf)''(\omega)\\
		&\quad \pm \frac{1}{2}\tau_j^2(1-\tau_j)^3\omega^3(\overline{Pf})^{'''}(\omega)
	\end{align*} 
	Rearranging:
	\begin{align*}
	&\omega (Pf_j)'(\omega) - \omega(Pf)'(\omega) = (-3\tau_j+3\tau_j^2-\tau_j^3)\omega(Pf)'(\omega) \\
	&\quad +  \tau_j(1-\tau_j)^3 \omega^2 (Pf)''(\omega) \pm \frac{1}{2}\tau_j^2(1-\tau_j)^3\omega^3(\overline{Pf})^{'''}(\omega) \, .
	\end{align*}
	Squaring and taking expectation:
	\begin{align*}
	&\Ex\left[\left(\omega g_j'(\omega) - \omega g'(\omega)\right)^2\right] \lesssim \left[\omega(Pf)'(\omega)\right]^2\eta^2 \\ &\qquad+\left[\omega^2(Pf)''(\omega)\right]^2\eta^2 + \left[\omega^3(\overline{Pf})^{'''}(\omega)\right]^2\eta^4 \, .
	\end{align*} 
	Having bounded $\text{var}\left[Z_j\right]$, an identical argument as the one used to control $\Ex \left[\norm[\big]{g_\eta-\widetilde{g}_\eta}_{2}^2\right]$ gives
	\begin{align*}
		\Ex &\left[ \norm[\big]{\omega g_\eta'(\omega)-\omega \widetilde{g}_\eta'(\omega)}_{2}^2 \right]
		\lesssim \frac{\eta^2}{M} \Big(\norm[\big]{\omega (Pf)'(\omega)}_2^2 \\
		&\qquad + \norm[\big]{\omega^2 (Pf)''(\omega)}_2^2 + \norm[\big]{\omega^3(\overline{Pf})^{'''}(\omega)}_2^2 \eta^2\Big)\, ,
	\end{align*}
	which proves the Theorem.
\end{proof}

Figure \ref{fig:compare_PS} illustrates how much is gained from inversion unbiasing for a specific high frequency signal; the mean power spectrum under Model \ref{model:dilMRA} is greatly perturbed due to large dilations, but $\widetilde{Pf}$ is still an accurate approximation of $Pf$. Although in general a signal is not uniquely defined by its power spectrum, if $\widehat{f}$ is real and positive as in Figure \ref{fig:SignalRecovery}, $f$ can be recovered from $Pf$. Figures \ref{fig:target_signal}--\ref{fig:mean_PS_rec_signal} illustrate how in this case inversion unbiasing yields a signal which accurately approximates the target.

\begin{figure}
	\centering
	\begin{subfigure}[b]{0.24\textwidth}
		\centering
		\includegraphics[width=.85\textwidth]{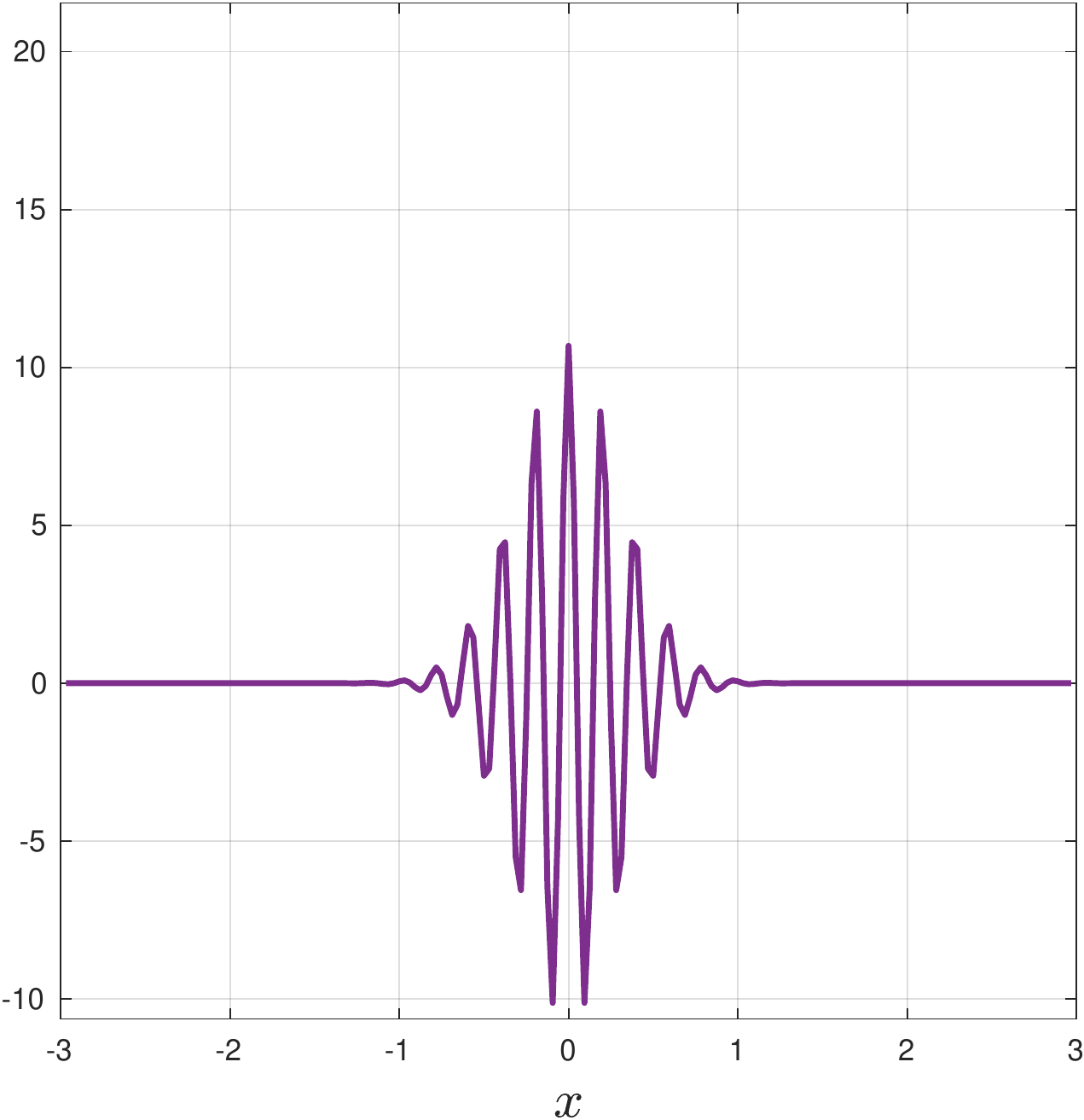}
		\caption{Target signal}
		\vspace*{.3cm}
		\label{fig:target_signal}
	\end{subfigure}
	\hfill
	\begin{subfigure}[b]{0.24\textwidth}
		\centering
		\includegraphics[width=.85\textwidth]{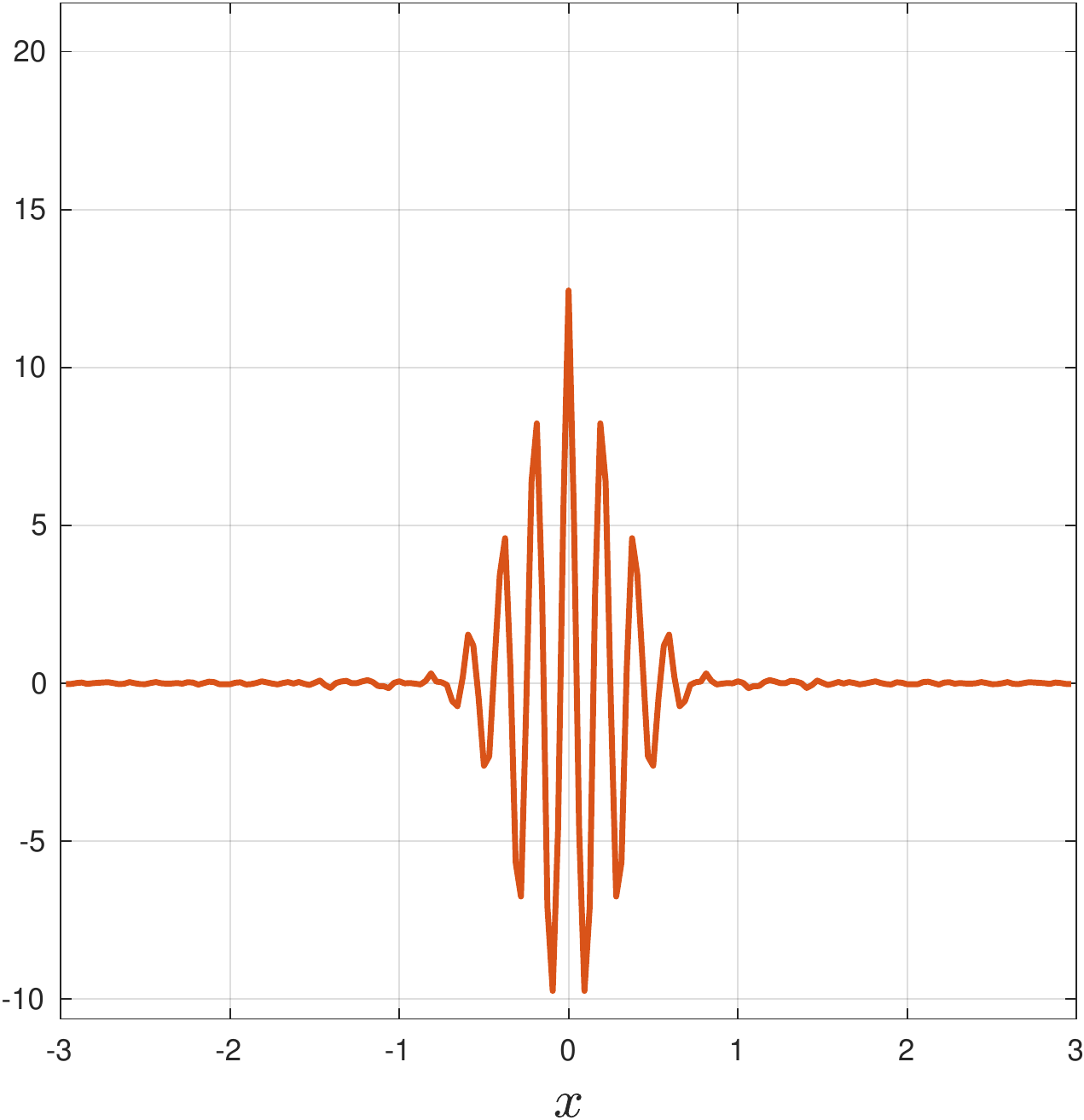}
		\caption{Signal recovered via $\widetilde{PS}$}
		\vspace*{.3cm}
		\label{fig:inv_unbiased_signal}
	\end{subfigure}
	\hfill
	\begin{subfigure}[b]{0.24\textwidth}
		\centering
		\includegraphics[width=.85\textwidth]{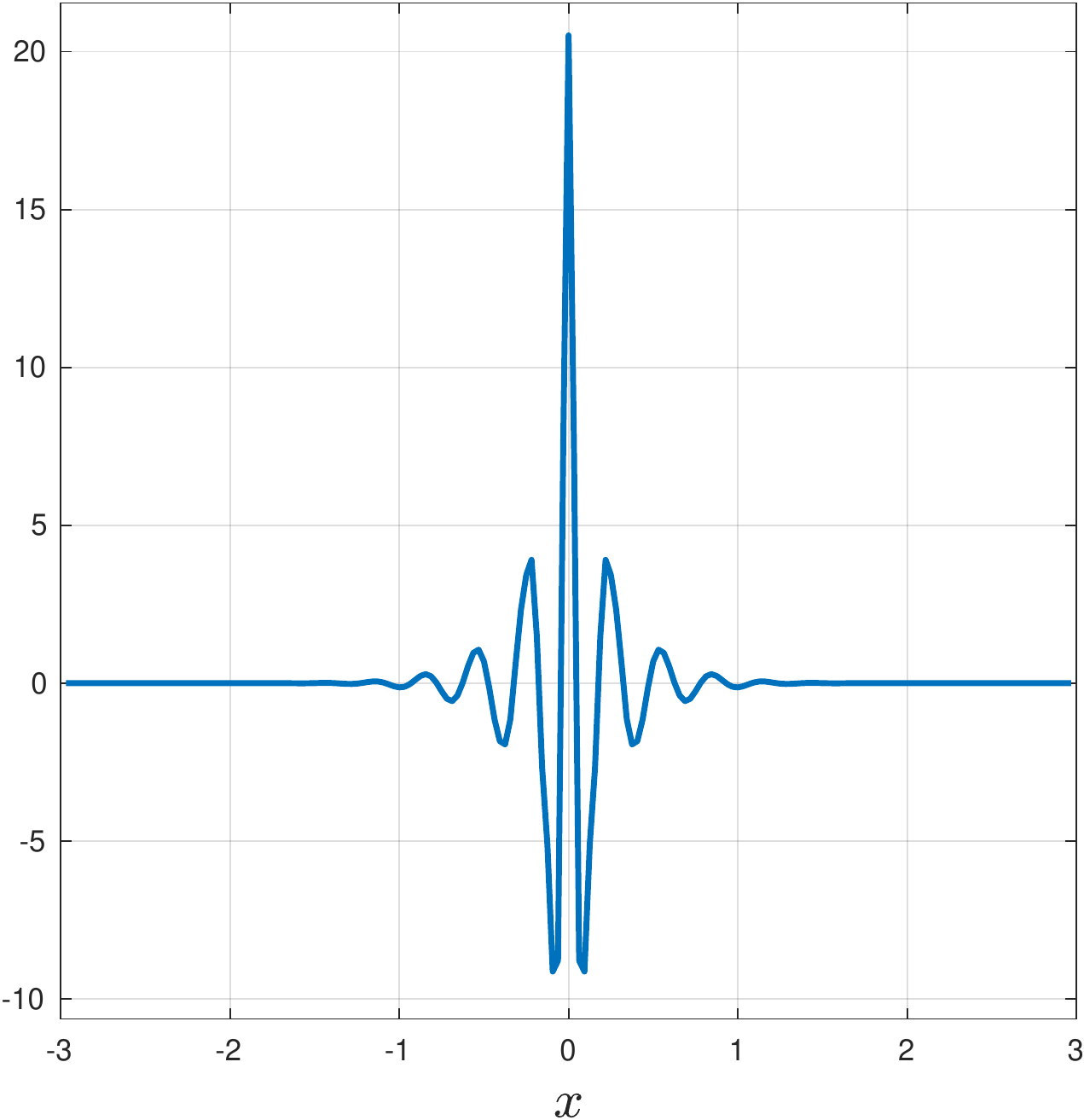}
		\caption{Signal recovered via mean PS}
		\label{fig:mean_PS_rec_signal}
	\end{subfigure}
	\begin{subfigure}[b]{0.24\textwidth}
		\centering
		\includegraphics[width=.85\textwidth]{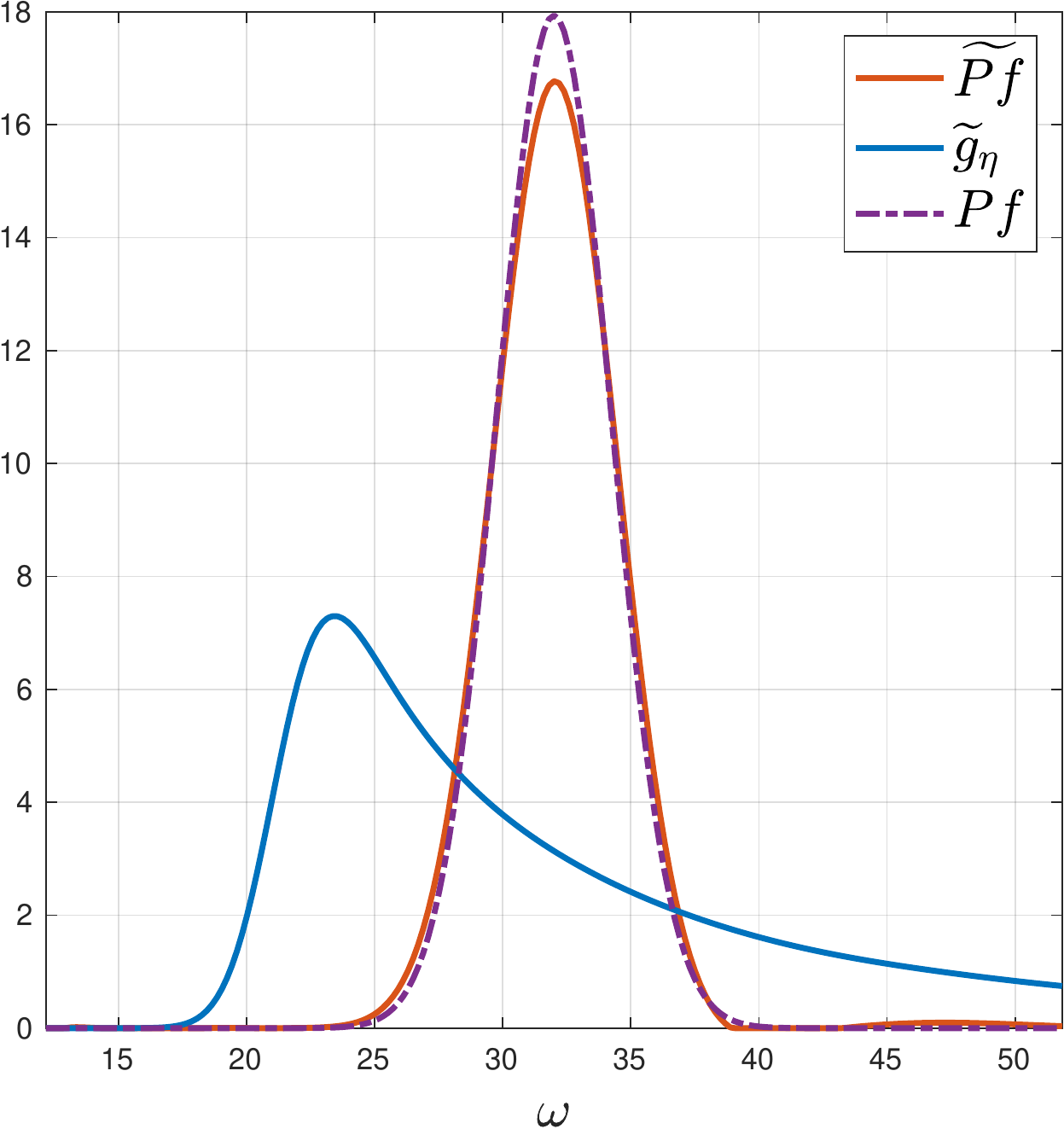}
		\caption{Power spectra}
		\label{fig:compare_PS}
	\end{subfigure}
	\caption{Power spectrum estimation and signal recovery for high frequency Gabor signal $f_3(x) = C_3 \exp^{-5x^2}\cos(32x)$ under Model \ref{model:dilMRA} with $\eta= 12^{-1/2}$ and $M=100,000$. The mean power spectrum $\widetilde{g}_\eta$ is greatly perturbed from the target power spectrum $Pf$, but applying inversion unbiasing to $\widetilde{g}_\eta$ yields an approximation $\widetilde{Pf}$ which is quite close to $Pf$ (see Figure \ref{fig:compare_PS}). Figure \ref{fig:target_signal} shows the target signal, and Figures \ref{fig:inv_unbiased_signal}, \ref{fig:mean_PS_rec_signal} show the target signal approximations obtained by inverting $\widetilde{g}_\eta$, $\widetilde{Pf}$.  }
	\label{fig:SignalRecovery}
\end{figure}

\subsection{Results for Noisy Dilation MRA}
\label{sec:NoisyDilationMRA}

Solving noisy dilation MRA presents several additional challenges which are lacking in dilation MRA.
First of all, the MSE can only be controled on a finite frequency interval due to the additive noise. We thus restrict to a finite frequency interval $\Omega$, 
and consider the MSE of an estimator $\widetilde{Pf}$ over the finite interval, i.e. $\Ex \left[\norm{Pf-\widetilde{Pf}}_{\Lb^2(\Omega)}^2\right]$. We note the residual error from working on $\Omega$ decays to zero as $|\Omega| \rightarrow \infty$. In addition, in any numerical implementation one is always restricted to a finite frequency interval. 

Another challenge is that one does not have direct access to $\widetilde{g}_\eta$; rather one only has access to 
\begin{align*}
	\frac{1}{M}\sum_{j=1}^M Py_j -\sigma^2 &= \widetilde{g}_\eta+ \widetilde{g}_\sigma
\end{align*}
where 
\begin{align*}
 \widetilde{g}_\sigma:= \frac{1}{M} \sum_{j=1}^M  \widehat{f}_{j}\widehat{\epsilon}_j^* +\widehat{f}_{j}^*\widehat{\epsilon}_j + P\epsilon_j- \sigma^2 \, .
\end{align*}
Although the compact support of the hidden signal guarantees the smoothness of $\widetilde{g}_\eta$, $\widetilde{g}_\sigma$ is not smooth due to the additive noise. To extend the unbiasing procedure of Section \ref{sec:DilationMRA} to the additive noise context, it is thus necessary to smooth the noisy power spectra. We thus compute $(\widetilde{g}_\eta+ \widetilde{g}_\sigma)\ast \phi_L$ where $\phi_L(\omega) = (2\pi L^2)^{-\frac{1}{2}}e^{-\frac{\omega^2}{2L^2}}$ is a Gaussian filter with width $L$, and then define the following estimator:
\begin{align}
	\label{equ:FiniteSamplePSest_genMRA}
	&(\widetilde{Pf})(\omega) := (I-L_{C_0})^{-1}C_1L_{C_2}\\
	&\quad\left[3 (\widetilde{g}_\eta+ \widetilde{g}_\sigma)\ast\phi_L(\omega)+\omega\left((\widetilde{g}_\eta+ \widetilde{g}_\sigma)\ast\phi_L\right)'(\omega)\right] \, . \nonumber
\end{align}
As $M\rightarrow \infty$ and $L\rightarrow 0$, \eqref{equ:FiniteSamplePSest_genMRA} is an unbiased estimator of $Pf$. To quantify how the error of the estimator depends on $L$, we need the following two lemmas.

\begin{lem}
	\label{lem:approx_identity_precise}
	Let $h \in \Lb^2 (\R)$ and assume $|\widehat{h} (\omega)|$ decays like $|\omega|^{-\alpha}$ for some integer $\alpha \geq 1$. Then for $L$ small enough:
	\begin{align*}
		\norm{ h - h \ast \phi_L }_2^2 &\lesssim \norm{ h }_2^2 L^4 + L^{4 \wedge (2\alpha - 1)} \, . 
	\end{align*}
\end{lem}

\begin{proof}
    The proof of Lemma \ref{lem:approx_identity_precise} is given in Appendix \ref{app: approx_identity_precise proof}. 
\end{proof}

\begin{lem}
	\label{lem:approx_identity_hw}
	Let $x h(x) \in\Lb^2 (\R)$ and assume $|\widehat{(\cdot)h(\cdot)} (\omega)|$ decays at least like $|\omega|^{-\alpha}$ for some integer $\alpha \geq 1$. Then for $L$ small enough:
	\begin{align*}
		&\| x(h - h \ast \phi_L) \|_2^2 \\
		&\quad\lesssim (L^3 \norm{h}_2^2) \wedge (L^4 \norm{h'}_2^2) + \|xh\|_2^2L^4 + L^{4\wedge(2\alpha-1)}\, .
	\end{align*}
\end{lem}

\begin{proof}
    The proof of Lemma \ref{lem:approx_identity_hw} is given in Appendix \ref{app: approx_identity_hw proof}. 
\end{proof}

We now state the main result of the article.
\begin{thm}
	\label{thm:main}
	Assume Model \ref{model:genMRA}, the estimator $(\widetilde{Pf})(\omega)$ defined in \eqref{equ:FiniteSamplePSest_genMRA}, $Pf \in \mathbf{C}^3 (\R)$, and that $\omega^k(\overline{Pf})^{(k)}(\omega) \in \Lb^2 (\R)$ for $k=2,3$.
	Then
	\begin{align*}
		&\Ex \left[\norm{Pf - \widetilde{Pf}}_{\Lb^2(\Omega)}^2\right] \lesssim C_{f, \Omega}\left(  \frac{\eta^2}{M} +L^4 + \frac{\sigma^2\vee\sigma^4}{L^2M}\right) \,.
	\end{align*} 
\end{thm}
\begin{proof}
	From Proposition \ref{prop:InfiniteSampleRecoveryPS} and a proof similar to Lemma \ref{lem:SimplifyError}
	\begin{align*}
		\norm{Pf - \widetilde{Pf}}_{\Lb^2(\Omega)}^2 &\lesssim \| g_\eta + \omega g_\eta'(\omega) - (\widetilde{g}_\eta+ \widetilde{g}_\sigma) \ast \phi_L \\
		&\qquad - \omega ((\widetilde{g}_\eta+ \widetilde{g}_\sigma) \ast \phi_L)'(\omega) \|_{\Lb^2(\Omega)}^2 \, .
	\end{align*}
	By the triangle inequality
	\begin{small}
	\begin{align*}
		&\norm{Pf - \widetilde{Pf}}_{\Lb^2(\Omega)}^2 \lesssim \norm{ g_\eta  + \omega g_\eta'(\omega) - \widetilde{g}_{\eta} -\omega \widetilde{g}_{\eta}'(\omega)}_{2}^2 \\
		&\ +\norm{\widetilde{g}_{\eta} + \omega \widetilde{g}_{\eta}'(\omega) -  (\widetilde{g}_\eta+ \widetilde{g}_\sigma) \ast \phi_L  - \omega ((\widetilde{g}_\eta+ \widetilde{g}_\sigma) \ast \phi_L)'(\omega)}_{\Lb^2(\Omega)}^2 \\
		&\ :=(A) +(B) \, .
	\end{align*}
\end{small}
	From the proof of Theorem \ref{thm:DilMRA_PS}, 
	\begin{align*}
		&\Ex[(A)] \lesssim \\
		&\ \ \frac{\eta^2}{M} \left(\norm{(Pf)(\omega)}_2^2+\norm{\omega(Pf)'(\omega)}_2^2 +\norm{\omega^2(Pf)''(\omega)}_2^2\right) + r\, ,
	\end{align*} 
	where $r = C_f \eta^4/M$ for a constant $C_f$ depending on $f$. It remains to control (B). We have
	\begin{align*}
		(B) 
		&\lesssim \norm{\widetilde{g}_{\eta} - \widetilde{g}_{\eta} \ast \phi_L}_2^2 + \norm{\omega\widetilde{g}_{\eta}' - \omega(\widetilde{g}_{\eta} \ast \phi_L)'}_2^2 \\
		&\quad + \norm{\widetilde{g}_{\sigma} \ast \phi_L}_{\Lb^2(\Omega)}^2 + \norm{\omega (\widetilde{g}_{\sigma}\ast\phi_L)'}_{\Lb^2(\Omega)}^2 \\
		&:=  (\text{I}) + (\text{II}) + (\text{III}) + (\text{IV})\, .
	\end{align*}
	We control (I) with Lemma \ref{lem:approx_identity_precise} and (II) with Lemma \ref{lem:approx_identity_hw}; we note in both cases $\alpha$ can be chosen arbitrarily large since the signals have compact support. 
	By Lemma \ref{lem:approx_identity_precise},
	\begin{align*}
		(\text{I}) &= \norm{\widetilde{g}_{\eta} - \widetilde{g}_{\eta} \ast \phi_L}_2^2 \lesssim L^4 \norm{\widetilde{g}_{\eta}}_2^2 \lesssim L^4 \norm{Pf}_2^2 \, ,
	\end{align*}
	since $\|Pf_j\|_2 =(1-\tau_j)^{\frac{3}{2}}\|Pf\|_2 \leq (\frac{3}{2})^{\frac{3}{2}}\|Pf\|_2$.
	By Lemma \ref{lem:approx_identity_hw},
	\begin{align*}
		(\text{II}) &= \norm{\omega\widetilde{g}_{\eta}' - \omega(\widetilde{g}_{\eta} \ast \phi_L)'}_2^2 \\
		&\lesssim L^4 \norm{\widetilde{g}_{\eta}''}_2^2 + L^4 \norm{\omega\widetilde{g}_{\eta}'(\omega)}_2^2 + L^4 \\
		&\lesssim  L^4 \left(\norm{(Pf)''}_2^2 + \norm{\omega(Pf)'(\omega)}_2^2 + 1\right)\, .
	\end{align*}
	For (III), note that by Young's Convolution Inequality
	\begin{align*}
		&\norm{\widetilde{g}_{\sigma} \ast \phi_L}_{\Lb^2(\Omega)}^2 \leq  \norm{\phi_L}_1^2 \cdot \norm{\widetilde{g}_{\sigma}}_{\Lb^2(\Omega)}^2 \\ &\quad=\norm{\widetilde{g}_{\sigma}}_{\Lb^2(\Omega)}^2 \\
		&\quad\lesssim \norm[\Big]{ \frac{1}{M} \sum_{j=1}^M  \widehat{f}_{j}\widehat{\epsilon}_j^*}_2^2 + \norm[\Big]{\frac{1}{M}\sum_{j=1}^M P\epsilon_j- \sigma^2}_{\Lb^2(\Omega)}^2 \, .
	\end{align*}
	We have
	\begin{align*}
		\Ex\left[\norm{ \frac{1}{M} \sum_{j=1}^M  \widehat{f}_{j}\widehat{\epsilon}_j^*}_2^2 \right] &= \int \Ex\left( \frac{1}{M} \sum_{j=1}^M  \widehat{f}_{j}(\omega)\widehat{\epsilon}_j^*(\omega)\right)^2\ d\omega \\
		&\leq \int \frac{1}{M^2} \sum_{j=1}^M \widehat{f}_{j}(\omega)^2 \sigma^2\ d\omega \\
		&\lesssim \frac{\sigma^2}{M}\norm{f}_2^2 \, .
	\end{align*}
	Since $\Ex[P\epsilon_j] = \sigma^2$, $\Ex[(P\epsilon_j)^2] \leq 3\sigma^4$ (see Lemma D.1 in \cite{hirn2020wavelet}), one has
	\begin{align*}
		\Ex\left(\frac{1}{M}\sum P\epsilon_j- \sigma^2\right)^2 &= \frac{\text{var}(P\epsilon_j)}{M} \leq \frac{3\sigma^4}{M} \, ,
	\end{align*}
	which implies
	\begin{align*}
	\Ex\left[\norm{\frac{1}{M}\sum P\epsilon_j- \sigma^2}_{\Lb^2(\Omega)}^2\right] &\lesssim |\Omega|\frac{\sigma^4}{M}\, .
	\end{align*}
	Thus
	\begin{align*}
		\Ex\left[ (\text{III})\right] &\lesssim   \frac{\sigma^2}{M}\left(\norm{f}_2^2+ |\Omega|\sigma^2\right)\, .
	\end{align*}
	For (IV), note that since $\norm{\phi_L'}_1^2 \sim L^{-2}$,
	\begin{align*}
		\norm{\omega (\widetilde{g}_{\sigma}\ast\phi_L)'}_{\Lb^2(\Omega)}^2 &\leq |\Omega|^2\, \norm{\widetilde{g}_{\sigma}\ast\phi_L'}_{\Lb^2(\Omega)}^2 \\
		&\leq |\Omega|^2\norm{\phi_L'}_1^2  \norm{\widetilde{g}_{\sigma}}_{\Lb^2(\Omega)}^2 \\
		&\lesssim \frac{|\Omega|^2}{L^2}\norm{\widetilde{g}_{\sigma}}_{\Lb^2(\Omega)}^2 \, ,
	\end{align*}
	so that utilizing our previous bound for $\Ex \left[\norm{\widetilde{g}_{\sigma}}_{\Lb^2(\Omega)}^2\right]$ one obtains
	\begin{align*}
		\Ex\left[ (\text{IV})\right] &\lesssim   \frac{|\Omega|^2\sigma^2}{L^2M}\left(\norm{f}_2^2+ |\Omega|\sigma^2\right) \, .
	\end{align*}
	Adding up the error terms:
	\begin{align*}
		&\norm{Pf - \widetilde{Pf}}_{\Lb^2(\Omega)}^2 \\
		&\quad\lesssim  \frac{\eta^2}{M} \left(\norm{(Pf)(\omega)}_2^2+\norm{\omega(Pf)'(\omega)}_2^2 +\norm{\omega^2(Pf)''(\omega)}_2^2\right) \\
		&\qquad+ r + L^4 \left(\norm{Pf}_2^2+\norm{(Pf)''}_2^2 + \norm{\omega(Pf)'(\omega)}_2^2+1\right) \\
		&\qquad+ \frac{|\Omega|^2\sigma^2}{L^2M}\left(\norm{f}_2^2+ |\Omega|\sigma^2\right) \\
		&\quad\lesssim C_{f, \Omega}\left(  \frac{\eta^2}{M} +L^4 + \frac{\sigma^2\vee\sigma^4}{L^2M}\right) \, ,
	\end{align*}
	which proves the theorem.
\end{proof}

To minimize the error upper bound in Theorem \ref{thm:main}, we balance the last two terms, i.e. we choose $L$ such that $L^4 \sim \frac{\sigma^2 \vee \sigma^4}{L^2M}$. In the high noise regime where $\sigma \geq 1$, this gives $L \sim \left(\frac{\sigma^4}{M}\right)^{\frac{1}{6}}$, which yields the following important corollary.

\begin{cor}
	\label{cor:error_for_best_L}
	Let the assumptions of Theorem \ref{thm:main} hold and in addition let $\sigma \geq 1$ and $L = \left(\frac{\sigma^4}{M}\right)^{\frac{1}{6}}$. Then:
	\begin{align*}
		\Ex\left[\norm{Pf - \widetilde{Pf}}_{\Lb^2(\Omega)}^2\right] &\lesssim C_{f, \Omega}\left[  \frac{\eta^2}{M} +\left(\frac{\sigma^4}{M}\right)^{\frac{2}{3}} \right] \,.
	\end{align*} 
\end{cor}

\begin{rmk}
	The inversion unbiasing procedure can also be directly applied to the wavelet-based features $(Sy)(\lambda)=\norm{y\ast \psi_{\lambda}}_2^2$, where $\psi_\lambda(x)=\sqrt{\lambda}\psi(\lambda x)$ is a wavelet with frequency $\lambda$, proposed in \cite{hirn2020wavelet} for solving Model \ref{model:genMRA}. Because these features are smooth by design, no additional smoothing is necessary, and when $\sigma\geq 1$ this will yield an estimator $\widetilde{Sf}$ with error
	\begin{align*}
		\Ex\left[\norm{Sf - \widetilde{Sf}}_{\Lb^2(\Omega)}^2\right] &\lesssim C_{f,\Omega}\left[  \frac{\eta^2}{M} +\frac{\sigma^4}{M} \right] \,.
	\end{align*} 
	The additive noise convergence rate for the wavelet-based features is slightly better than the convergence rate for the power spectrum given in Corollary \ref{cor:error_for_best_L}.
	A power spectrum estimator $\widetilde{Pf}$ can then be obtained from $\widetilde{Sf}$, since the wavelet-based features are defined by an invertible operator on the power spectrum. However, this inversion process is highly unstable numerically, as its accuracy is governed by the smallest eigenvalue of a low rank matrix. In practice, applying inversion unbiasing directly to the power spectrum yielded a lower error in our numerical experiments.  
\end{rmk}

\section{Optimization}
\label{sec:opt}

To actually compute the estimator \eqref{equ:FiniteSamplePSest_genMRA}, one must apply the inverse operator $(I-{L_{C_0}})^{-1}$. A simple formula for this inversion is unavailable; however it is straightforward to compute the estimators by solving a convex optimization problem. In the infinite sample limit, one has access to the perfect data term
\begin{align*}
	d(\omega) = 3g_\eta(\omega)+\omega g_\eta'(\omega) \, ,
\end{align*}   
and Proposition \ref{prop:InfiniteSampleRecoveryPS} guarantees that $g=Pf$ can be recovered from $d$ by
\begin{align*}
	g &= \argmin_{\varg}\ \norm[\big]{(I-L_{C_0})\varg -C_1L_{C_2}d }_2^2 \, ,
\end{align*}
where the constants $C_i$ depend on $\eta$. In practice the variation parameter $\eta$ may be unknown, so the relevant loss function is: 
\begin{align*}
	\mathcal{L}(\varg,\vareta) &= \norm[\big]{ \left(I-L_{C_0(\vareta)}\right)\varg - C_1(\vareta)L_{C_2(\vareta)}d}_2^2 \, .
\end{align*}
The following Proposition guarantees that the infinite sample loss function $\mathcal{L}$ has a unique critical point, and thus that $g=Pf$ can be recovered by minimizing $\mathcal{L}$. 

\begin{prop}
	\label{prop:uniqueCP}
	Let $g \in \Cb^2(\mathbb{R}), \eta>0$ be the true power spectrum and dilation standard deviation, and assume $g(0)\ne 0, g''(0) \ne 0$. Then $(g,\eta)$ is the only critical point of $\mathcal{L}(\varg,\vareta)$ in $(\Cb^2(\mathbb{R}), \R^{+})$.
\end{prop}
\begin{proof}
	
	We first compute $\nabla_{\varg} \mathcal{L}(\varg,\vareta)$ and $\nabla_{\vareta} \mathcal{L}(\varg,\vareta)$. To compute $\nabla_{\varg} \mathcal{L}(\varg,\vareta)$, we first view $\vareta$ as fixed, and compute the Frechet derivative of $\mathcal{L}(\varg)$. Let $A = I-L_{C_0}$; throughout the proof, $A$ and the constants $C_i$ depend on $\vareta$ but for brevity we do not explicitly denote this dependence. Note
	\begin{align*}
		\mathcal{L}(\varg) &= \norm{A \varg - C_1L_{C_2}d}_2^2 = N(A \varg) \, ,
	\end{align*} 
	where $Nf = \norm{f-C_1L_{C_2}d}_2^2$.
	Thus by the chain rule, the functional derivative at $\varg$ applied to a test function $h$ is
	\begin{align*}
		(D\mathcal{L})(\varg)h &= (DN)(A \varg) \circ D(A \varg)h=  (DN)(A \varg) \circ A h
	\end{align*}
	since $A$ is a linear operator. To compute $DN$, note that
	\begin{align*}
		\frac{|N(f+h)-Nf-2\langle f-C_1L_{C_2}d,h\rangle|}{\norm{h}_2} &= \frac{\norm{h}_2^2}{\norm{h}_2} \rightarrow 0
	\end{align*}
	as $\norm{h}_2 \rightarrow 0 $,
	so $(DN)(f)h = 2\langle f-C_1L_{C_2}d, h\rangle$. Thus 
	\begin{align*}
		(D\mathcal{L})(\varg)h &= 2\langle A \varg-C_1L_{C_2}d, A h\rangle \\
		&= \langle 2A^*(A \varg-C_1L_{C_2}d),  h\rangle \\
		\implies \nabla \mathcal{L}(\varg) &= 2A^*(A \varg-C_1L_{C_2}d) \, .
	\end{align*}
	We thus have 
	\begin{align*}
		&\nabla_{\varg} \mathcal{L}(\varg,\vareta)  = 2A^*(A \varg-C_1L_{C_2}d) \\
		&\nabla_{\vareta} \mathcal{L}(\varg,\vareta) = \\
		&\ \int 2(A \varg(\omega) - C_1L_{C_2}d(\omega))\frac{\partial}{\partial \vareta}(A \varg(\omega) - C_1L_{C_2}d(\omega)) \ d\omega \, .
	\end{align*}
	
	Since as demonstrated in the previous section $A g =C_1L_{C_2}d$ when $\vareta=\eta$, $\nabla_{\varg} \mathcal{L}(g,\eta) =\nabla_{\vareta} \mathcal{L}(g,\eta)=0$, and $(g,\eta)$ is a critical point of $\mathcal{L}$. We now show $(g,\eta)$ is the only critical point. 
	
	Assume $(\varg,\vareta)$ is a critical point. Then $2A^*(A \varg-C_1L_{C_2}d)=0$.  Since $C_0<1$, $A=I-L_{C_0}$ is invertible as was previously argued; thus its adjoint $A^*$ is also invertible, and $A \varg = C_1L_{C_2}d$ in $\Lb^2$. Since $L_{C_2}$ is a dilation operator and thus invertible, $B_{\vareta} \varg = d$ in $\Lb^2$, where $B_{\vareta} = C_1^{-1}L_{C_2}^{-1}A = C_1^{-1}L_{C_2}^{-1}(I-L_{C_0})$.
	Next we show that if $B_{\vareta} \varg = B_{\eta}g$ in $\Lb^2$, we must have $(\varg,{\vareta})=(g,\eta)$.
	It is easy to check from our definition of $C_0,C_1,C_2$ that
	\begin{align*}
		(B_{\vareta}\varg)(\omega) &= \frac{(1+\sqrt{3}\vareta)^3}{2\sqrt{3}\vareta} \varg\left((1+\sqrt{3}\vareta)\omega\right) \\
		&\quad - \frac{(1-\sqrt{3}\vareta)^3}{2\sqrt{3}\vareta} \varg\left((1-\sqrt{3}\vareta)\omega\right) \, .
	\end{align*}
	Note that
	\begin{align*}
		(B_{\vareta}\varg)(0) &=\frac{1}{2\sqrt{3}\vareta}\left((1+\sqrt{3}\vareta)^3-(1-\sqrt{3}\vareta)^3\right)\varg(0) \\
		&= \left(3+3\vareta^2\right)\varg(0)
	\end{align*}
	and similarly for $(B_{\eta}g)(0)$. Since the functions are equal in $\Lb^2$ and continuous, we must have
	\begin{align*}
		\left(3+3\vareta^2\right)\varg(0) &= \left[3+3\eta^2\right]g(0) \, .
	\end{align*}
	In addition $(B_{\vareta}\varg)''(0)$ satisfies
	\begin{align*}
		(B_{\vareta}\varg)''(0) &= \frac{1}{2\sqrt{3}\vareta}\left((1+\sqrt{3}\vareta)^5-(1-\sqrt{3}\vareta)^5\right)\varg''(0) \\
		&=(5+30\vareta^2+9\vareta^4)\varg''(0)
	\end{align*}
	and similarly for $(B_{\eta}g)''(0)$. Again since the functions are equal in $\Lb^2$  and continuously differentiable, we must have
	\begin{align*}
		(5+30\vareta^2+9\vareta^4)\varg''(0) &= [5+30(\eta)^2+9(\eta)^4]g''(0) \, .
	\end{align*}
	So
	\begin{align*}
		\varg(0) &= K_1 g(0) \\
		\varg''(0) &= K_2 g''(0) 
	\end{align*}
	for constants $K_1,K_2>0$ depending on $\vareta,\eta$. We conclude we must have $K_1=K_2$. So
	\begin{align*}
		K_1 = K_2
		\iff& \frac{3+3(\eta)^2}{3+3\vareta^2} = \frac{5+30(\eta)^2+9(\eta)^4}{5+30\vareta^2+9\vareta^4} \\
		\iff& \left[3+3(\eta)^2\right](5+30\vareta^2+9\vareta^4) \\
		&\quad = (3+3\vareta^2)[5+30(\eta)^2+9(\eta)^4] \, .
	\end{align*}
	Since $\vareta=\eta$ is the only real, positive solution of the above, we conclude $\vareta=\eta$. Thus $B_{\eta}g = B_{\eta}\varg $. Since $B_{\eta}$ is invertible, we conclude that $\varg=g$, and the Proposition is proved.
\end{proof}
In practice one only has access to the finite sample data term and loss function:
\begin{align*}
	\widetilde{d}(\omega) &:= 3 (\widetilde{g}_\eta+ \widetilde{g}_\sigma)\ast\phi_L(\omega)+\omega\left[(\widetilde{g}_\eta+ \widetilde{g}_\sigma)\ast\phi_L'\right](\omega) \\ 
	\widetilde{\mathcal{L}}(\varg,\vareta) &:= \norm[\big]{ \left(I-L_{C_0(\vareta)}\right)\varg - C_1(\vareta)L_{C_2(\vareta)}\widetilde{d}}_2^2\, ,
\end{align*}
and the estimator \eqref{equ:FiniteSamplePSest_genMRA} is computed by minimizing $\widetilde{\mathcal{L}}$. However, as $M \rightarrow \infty$, Proposition \ref{prop:uniqueCP} guarantees the optimization procedure has a unique critical point and is thus well behaved. However for finite $M$, the optimization can be delicate: since  $\widetilde{\mathcal{L}}(\varg, 0)=0$ for any $\varg$, there is a large plateau defined by $\eta=0$ were loss values are small even for $\varg$ very far from $Pf$. It thus becomes necessary to constrain $\eta$ to be bounded away from 0; Section \ref{sec:simu_results} describes specific implementation details.

\begin{rmk}
	If $\eta$ is known so the optimization is just over $g$, the optimization is convex.
\end{rmk}	

\begin{rmk}
	 In practice we define $\varrootg=\sqrt{\varg}$, optimize over $\varrootg$ to obtain the optimal $p$, and then define $g = p^2$; such a procedure ensures $g$ is nonnegative without constraining $\varg$ in the optimization. Note to implement the minimization of $\widetilde{\mathcal{L}}(\varg,\vareta)$, one needs to compute $A^*$ for the operator $A=I-L_{C_0}$. A straightforward calculation shows $A^*h(\omega) = h(\omega) - C_0^2 h\left(\frac{\omega}{C_0}\right)$.
\end{rmk}

\section{Simulation Results}
\label{sec:simu_results}

\begin{figure}
	\centering
	\begin{subfigure}[b]{0.24\textwidth}
		\centering
		\includegraphics[width=.85\textwidth]{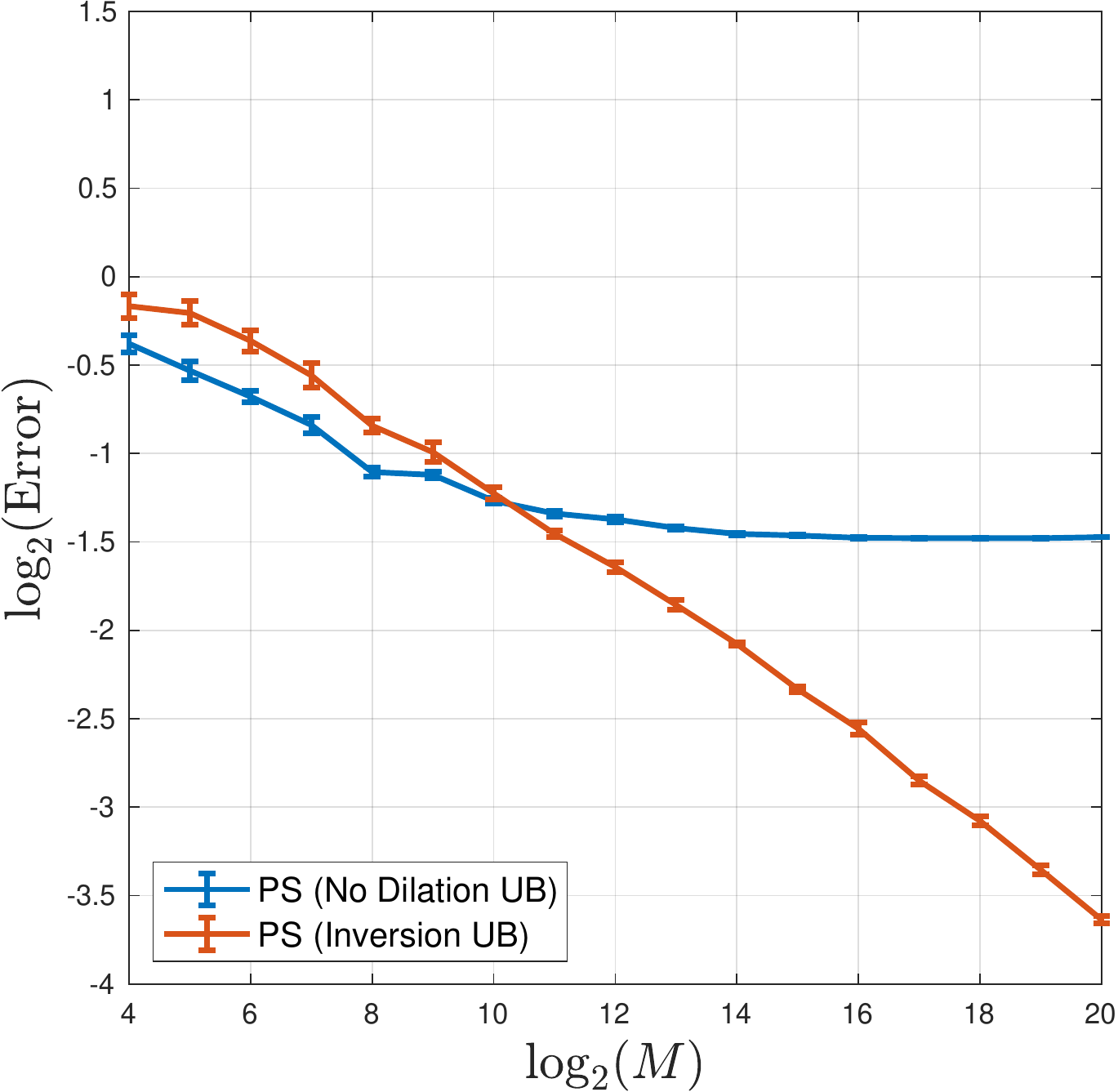}
		\caption{$f_1$ (slope $= -0.2498$)}
		\vspace*{.1cm}
		\label{fig:oracle_f1}
	\end{subfigure}
	\hfill
	\begin{subfigure}[b]{0.24\textwidth}
		\centering
		\includegraphics[width=.85\textwidth]{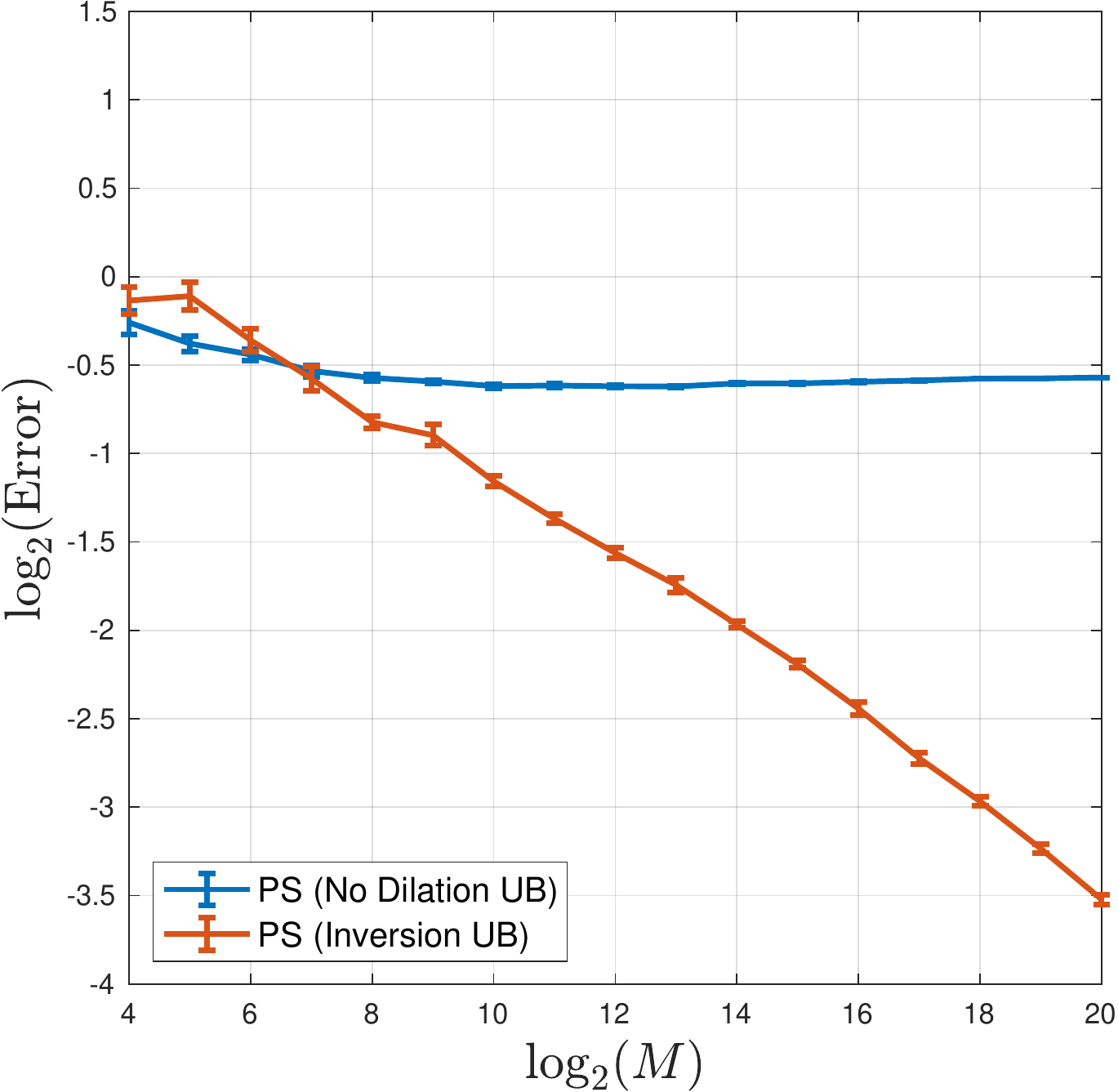}
		\caption{$f_2$ (slope $=-0.2473$)}
		\vspace*{.1cm}
		\label{fig:oracle_f2}
	\end{subfigure}
	\hfill
	\begin{subfigure}[b]{0.24\textwidth}
		\centering
		\includegraphics[width=.85\textwidth]{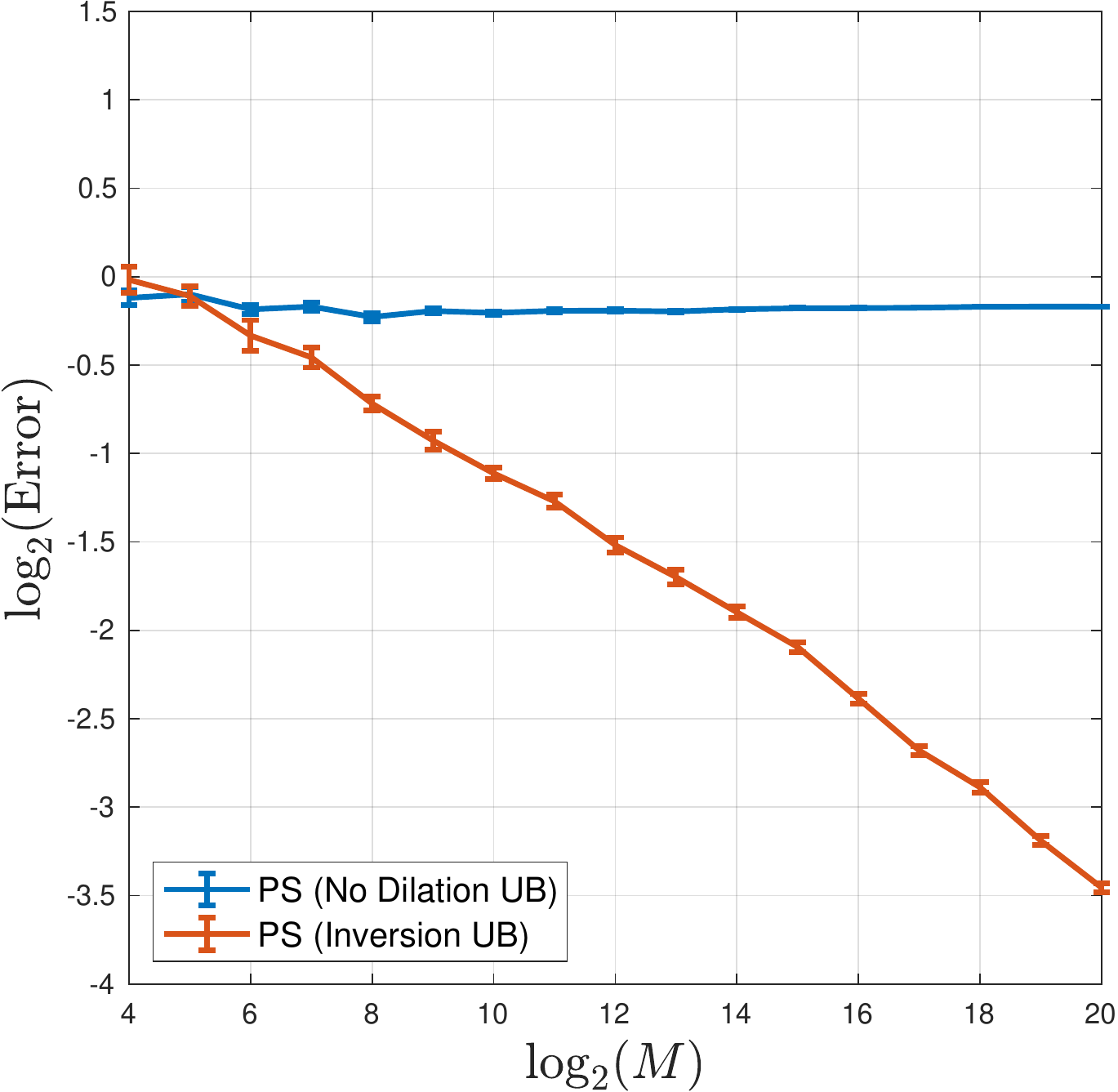}
		\caption{$f_3$ (slope $=-0.2464$)}
		\vspace*{.1cm}
		\label{fig:oracle_f3}
	\end{subfigure}
	\begin{subfigure}[b]{0.24\textwidth}
		\centering
		\includegraphics[width=.85\textwidth]{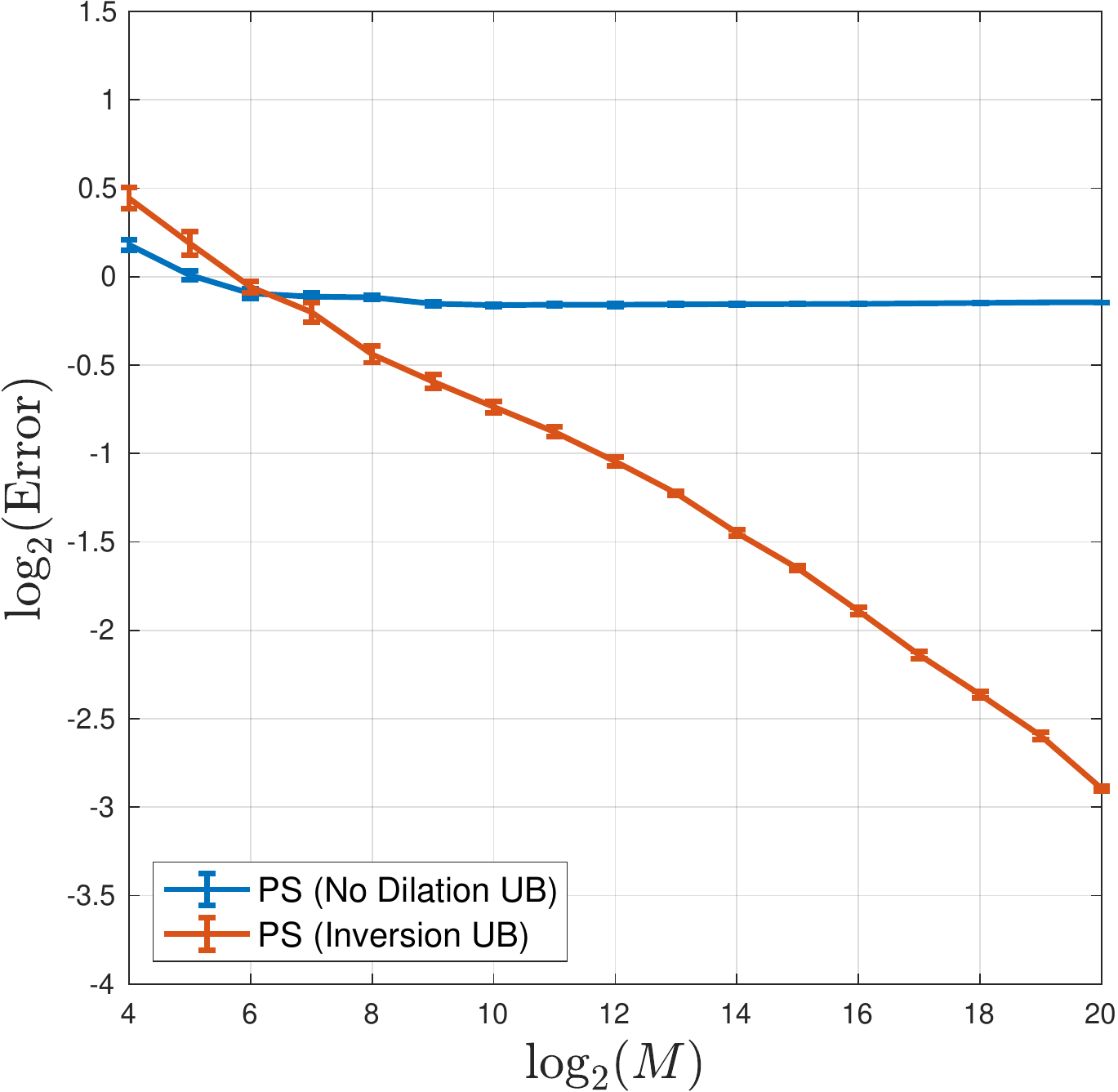}
		\caption{$f_4$ (slope $ = -0.2306$)}
		\vspace*{.1cm}
		\label{fig:oracle_f4}
	\end{subfigure}
	\hfill
	\begin{subfigure}[b]{0.24\textwidth}
		\centering
		\includegraphics[width=.85\textwidth]{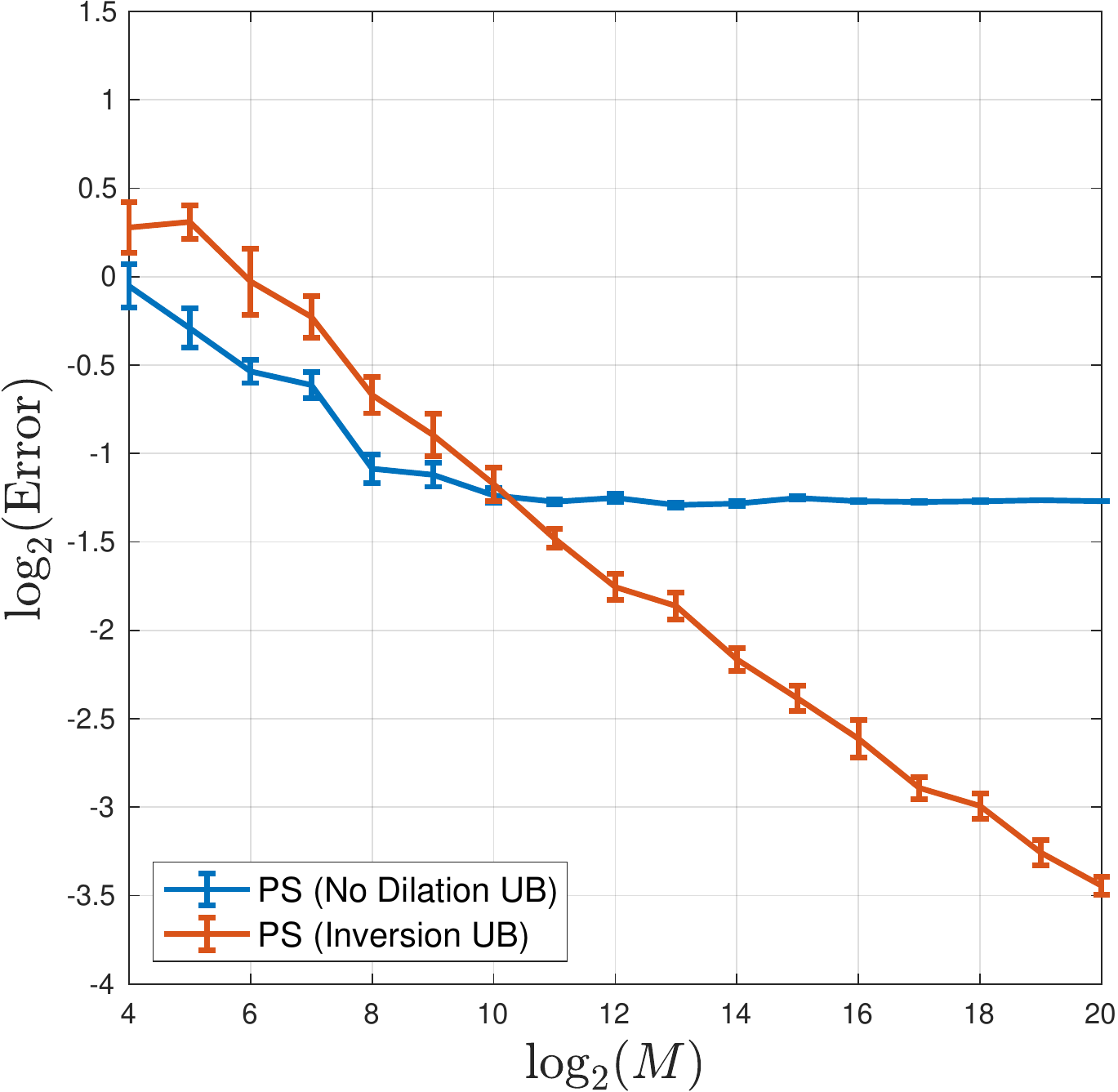}
		\caption{$f_5$ (slope $= -0.2184$)}
		\vspace*{.1cm}
		\label{fig:oracle_f5}
	\end{subfigure}
	\hfill
	\begin{subfigure}[b]{0.24\textwidth}
		\centering
		\includegraphics[width=.85\textwidth]{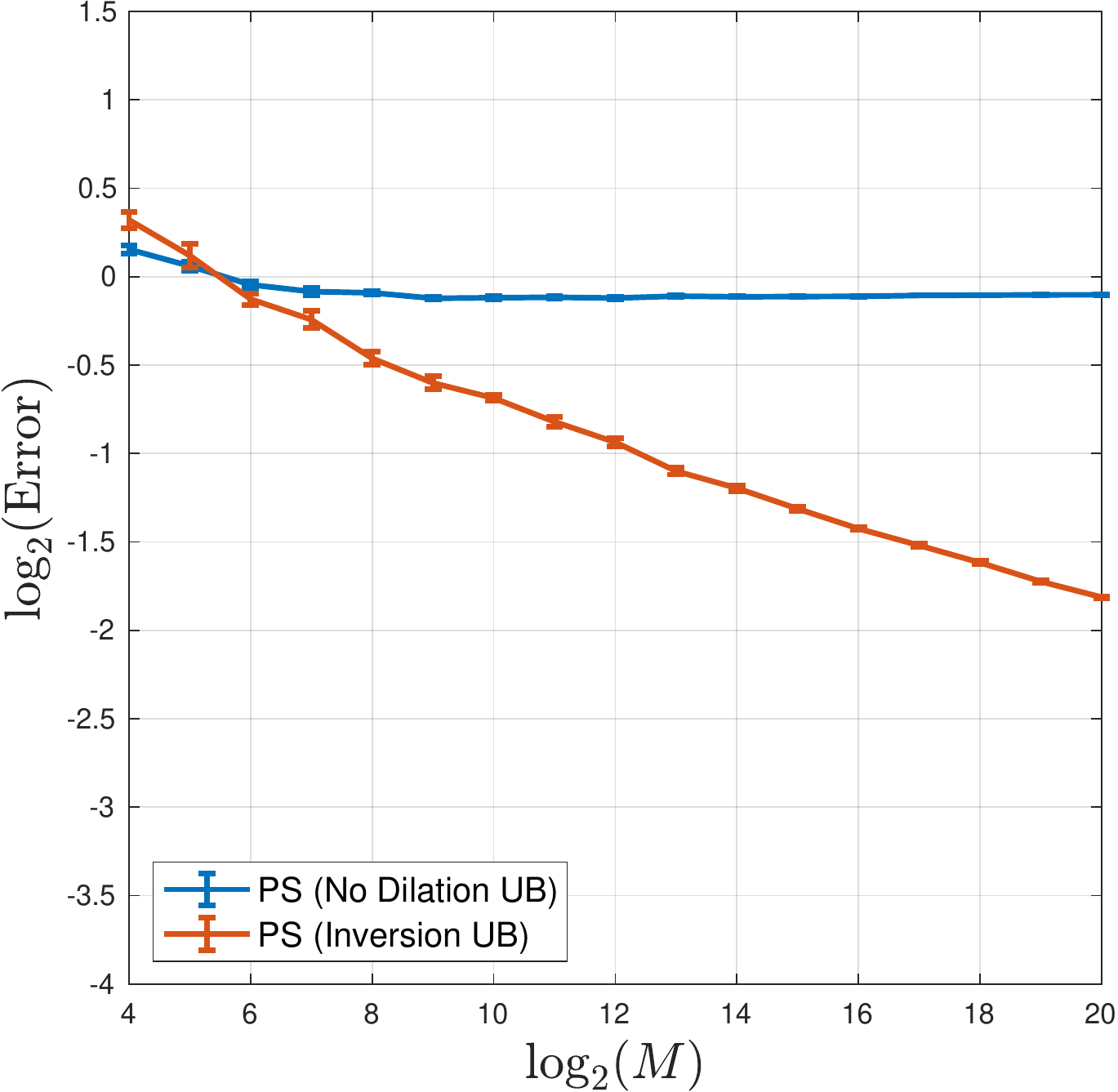}
		\caption{$f_6$ (slope $= -0.1071$)}
		\vspace*{.1cm}
		\label{fig:oracle_f6}
	\end{subfigure}
	\begin{subfigure}[b]{0.24\textwidth}
		\centering
		\includegraphics[width=.85\textwidth]{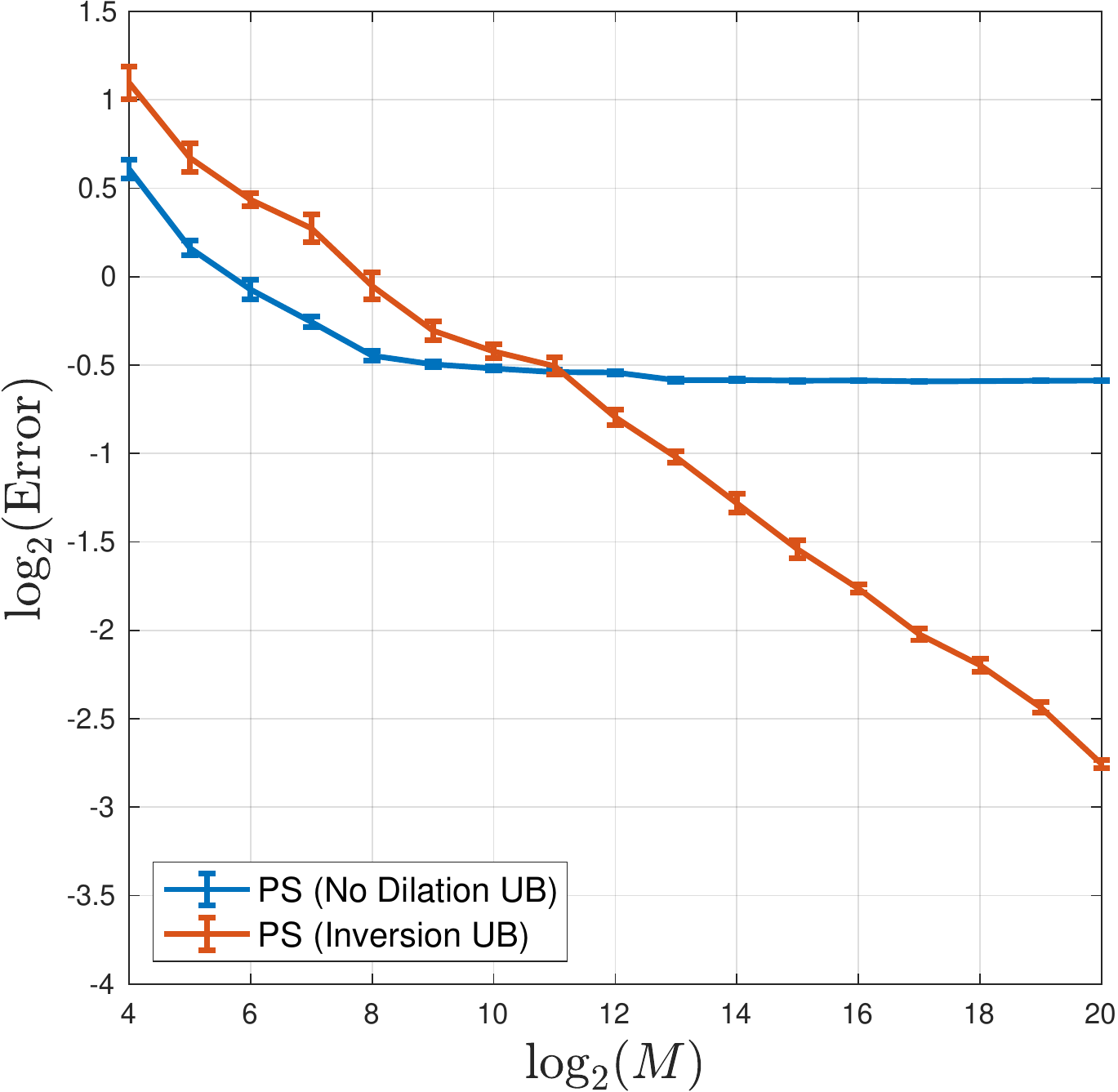}
		\caption{$f_7$ (slope $= -0.2400$)}
		\vspace*{.1cm}
		\label{fig:oracle_f7}
	\end{subfigure}
	\hfill
	\begin{subfigure}[b]{0.24\textwidth}
		\centering
		\includegraphics[width=.85\textwidth]{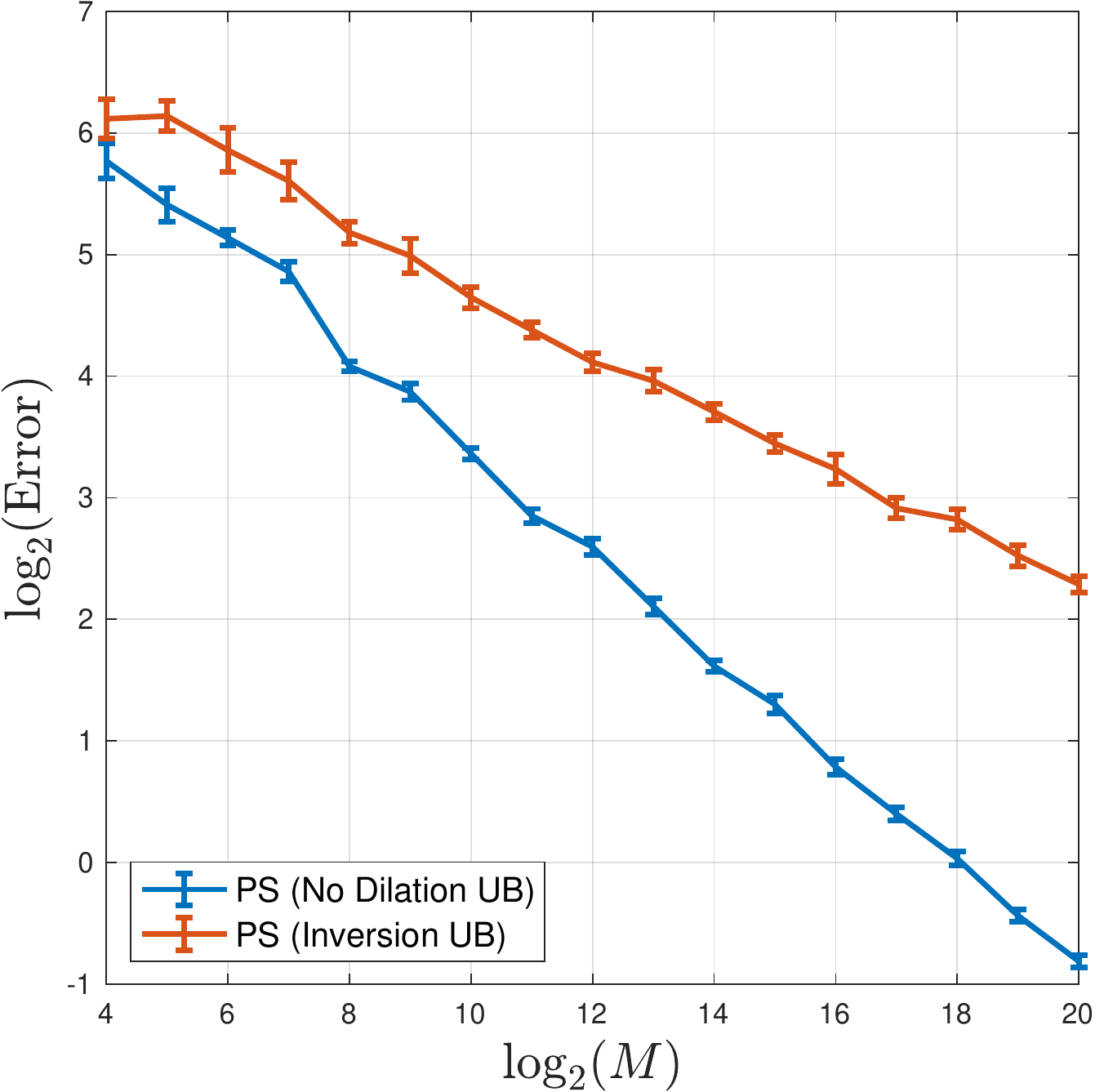}
		\caption{$f_8$ (slope $= -0.2320$)}
		\vspace*{.1cm}
		\label{fig:oracle_f8}
	\end{subfigure}
	\caption{Error decay with standard error bars for Model \ref{model:genMRA} (oracle moment estimation).  All plots show relative $\Lb^2$ error and have the same axis limits, except Figure \ref{fig:oracle_f8}, which shows absolute error. Reported slopes were computed by linear regression on the right half of the plot, i.e. for $12\leq \log_2(M)\leq 20$.}
	\label{fig:GenMRAModelOracle}
\end{figure}

In this section we investigate the proposed inversion unbiasing procedure on the following collection of synthetic signals which capture a variety of features:
\begin{align*}
	f_1(x) &= C_1 \exp^{-5x^2}\cos(8x) \\
	f_2(x) &= C_2 \exp^{-5x^2}\cos(16x)\\
	f_3(x) &= C_3 \exp^{-5x^2}\cos(32x)\\
	\widehat{f}_4(\omega) &= C_4 \left[\text{sinc}(0.2(\omega-32))+\text{sinc}(0.2(-\omega-32)) \right] \\
	f_5(x) &= C_5 \exp^{-0.04x^2}\cos(30x+1.5x^2) \\
	\widehat{f}_6(\omega) &= C_6 \left[\mathbf{1}(\omega\in[-38,-32])+\mathbf{1}(\omega\in[32,38]) \right] \\
	\widehat{f}_7(\omega) &= C_7 \left[\text{zigzag}\left(0.2(\omega+40)\right)+\text{zigzag}\left(0.2(\omega+40)\right)\right]^{1/2}\\
	f_8(x) &= 0  
\end{align*}
The hidden signals were defined on $[-\frac{N}{4}, \frac{N}{4}]$ and the corresponding noisy signals on $[-\frac{N}{2}, \frac{N}{2}]$. The signals were sampled at rate $1/2^{\ell}$, resolving frequencies in the interval $[-2^{\ell}\pi, 2^{\ell}\pi]$; $N=2^5$ and $\ell =5$ were used for all simulations. 
As indicated above, $f_4, f_6, f_7$ were sampled directly in the frequency domain, while the rest were sampled in the spatial domain. The normalization constants $C_i$ were chosen so that all signals would have the same $\snr$ for a fixed additive noise level, specifically $(\snr)^{-1} = \sigma^2$, where $\snr = \left(\frac{1}{N} \int_{-N/2}^{N/2} f(x)^2\ dx\right) / \sigma^2$. 

\begin{figure}
	\centering
	\begin{subfigure}[b]{0.24\textwidth}
		\centering
		\includegraphics[width=.85\textwidth]{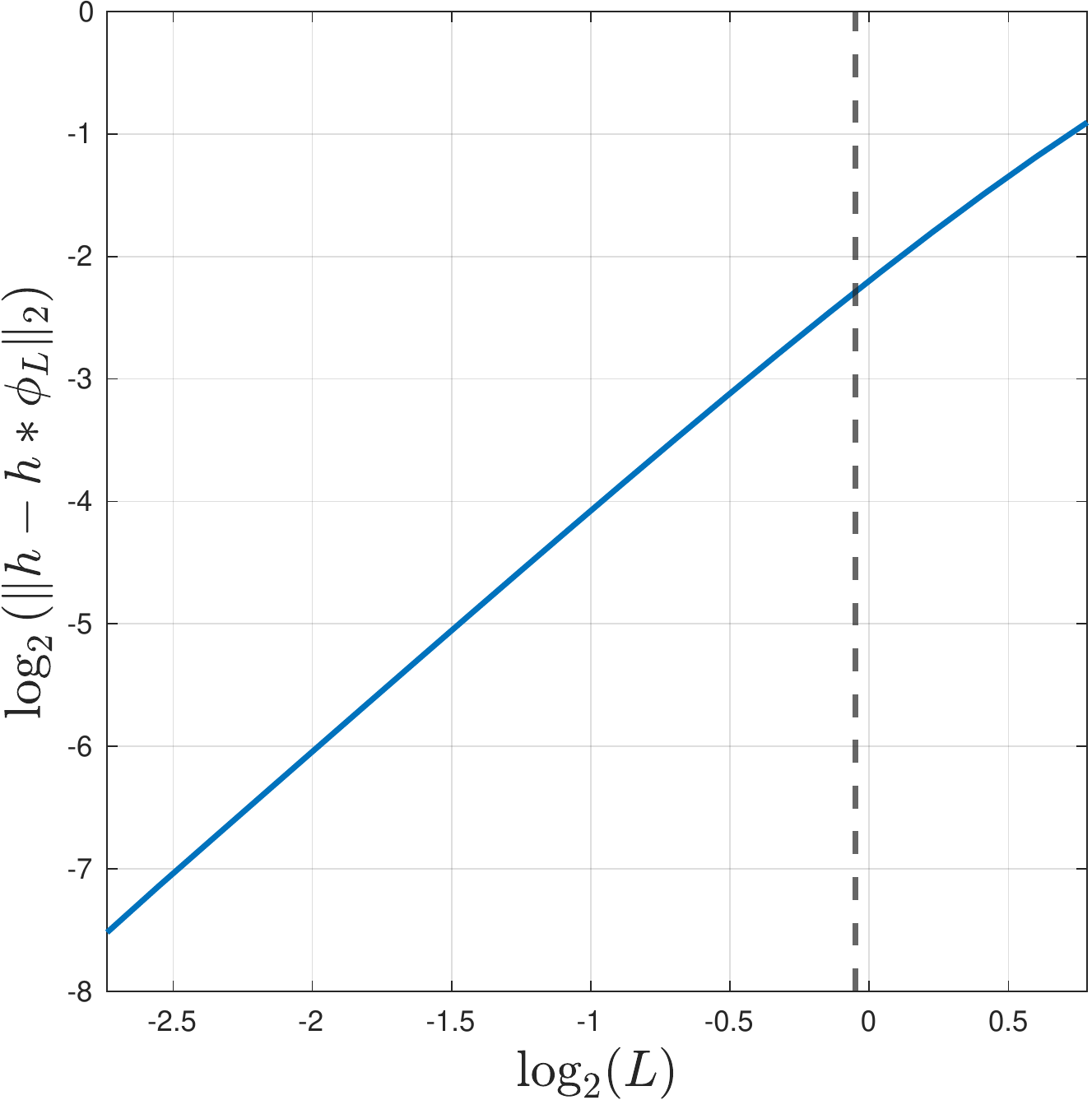}
		\caption{Smoothing decay rate}
		\vspace*{.1cm}
		\label{fig:Smoothing}
	\end{subfigure}
	\hfill
	\begin{subfigure}[b]{0.24\textwidth}
		\centering
		\includegraphics[width=.8\textwidth]{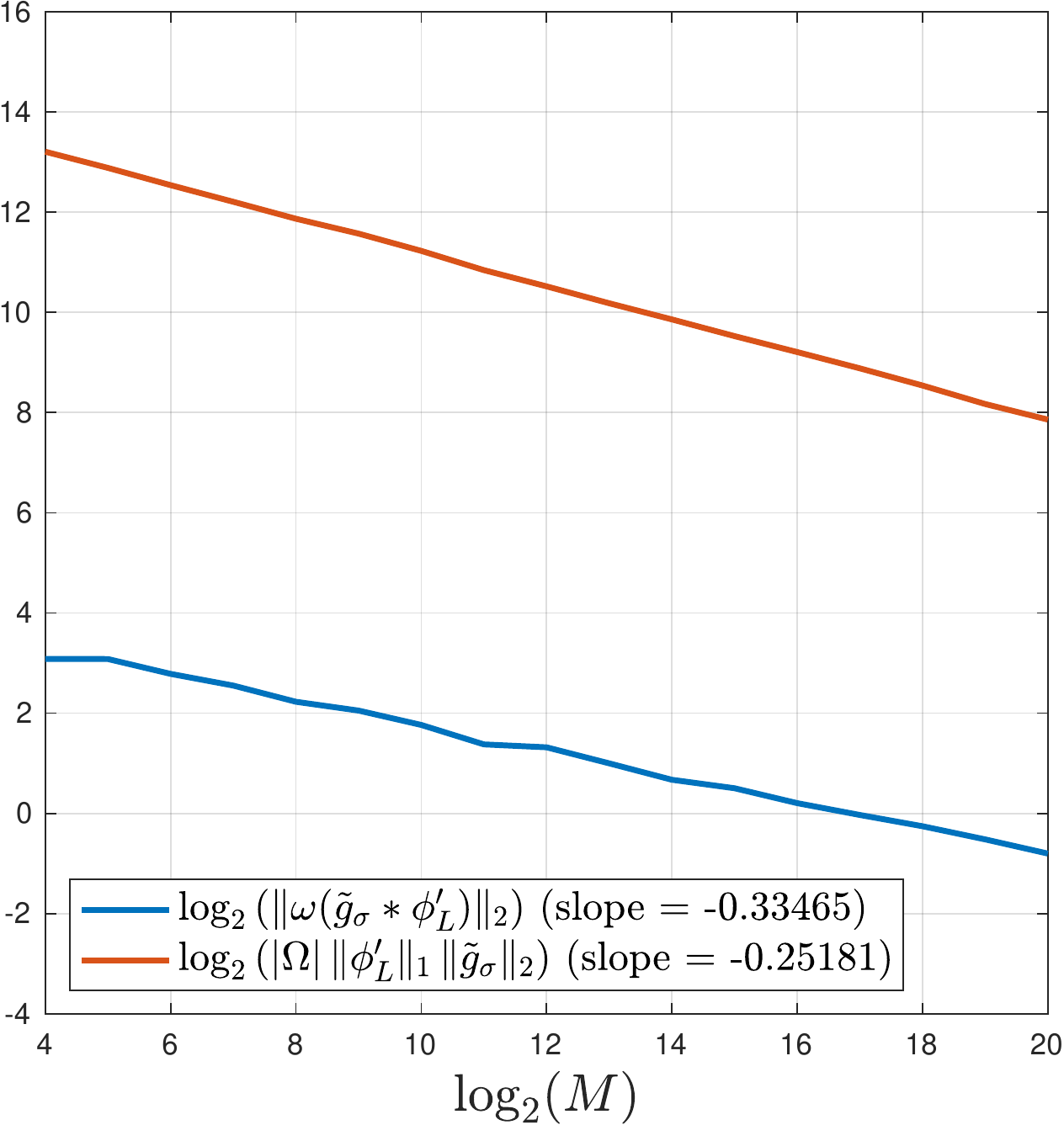}
		\caption{Young's Inequality}
		\vspace*{.1cm}
		\label{fig:Young}
	\end{subfigure}
	\caption{Plots explaining the small discrepancy between theoretical and empirical convergence rates. Figure \ref{fig:Smoothing}: The right side of the dashed line shows the $L$ values corresponding to $12 \leq \log_2(M) \leq 20$, i.e. the upper range of values used in our simulations. In the simulation regime, the slope in the log-log plot is 1.65; however for small $L$ (left side of dashed line), the slope is 1.96, which closely matches the $L^2$ rate given in Lemma \ref{lem:approx_identity_hw}. Figure \ref{fig:Young}:  the additive noise term exhibits a decay rate of -0.25 in the range of $M$ values used for our simulations, while the upper bound due to Young's Inequality decays at the faster rate of -0.33.}
	\label{fig:Discrepancy}
\end{figure}

The Gabors $f_1-f_3$ are smooth with a fast decay in both space and frequency; $f_4$ is discontinuous in space, with a smooth but slowly decaying FT; $f_5$ is a linear chirp with a non-constant instantaneous frequency; $f_6$ is discontinuous in frequency; $f_7$ is continuous but not smooth in frequency. The zero signal was included to investigate the effect of the inversion unbiasing procedure when applied directly to additive noise, i.e. in the absence of any signal. We investigate the ability of inversion unbiasing to solve Model \ref{model:genMRA} in the challenging regime of both low $\snr$ and large dilations. Specifically we choose $\snr = \frac{1}{2}$ and $\tau$ uniform on $[-\frac{1}{2}, \frac{1}{2}]$ (thus $\sigma = \sqrt{2}$ and $\eta= 12^{-1/2} \approx 0.2887$). For comparison, the simulations in \cite{hirn2020wavelet} were restricted to $\eta \leq 0.12$.

We first assume oracle knowledge of the additive noise and dilation variances $\sigma^2, \eta^2$. We let $M$ increase exponentially from 16 to $1,048,576$, and for each value of $M$ we run 10 simulations of Model \ref{model:genMRA} and compute $\widetilde{Pf}$ as given in \eqref{equ:FiniteSamplePSest_genMRA}. The width of the Gaussian filter $L$ is chosen as in Corollary \ref{cor:error_for_best_L}, and the inversion operator is applied by solving a convex optimization problem as described in Section \ref{sec:opt}.  For each simulation, the relative error of the resulting power spectrum estimator is computed as 
\begin{align*}
	\text{Error} := \frac{\norm{Pf -\widetilde{Pf}}_2}{\norm{Pf}_2} \, ,
\end{align*}
and the mean error is then computed across simulations. Figure \ref{fig:GenMRAModelOracle} shows the decay of the mean error as the sample size $M$ increases. All signals exhibit a linear error decay in the log-log plots; as the error decay does not plateau, the simulations confirm that $\widetilde{Pf}$ is an unbiased estimator of $Pf$ as shown in Theorem \ref{thm:main} and Corollary\ref{cor:error_for_best_L}. 

\begin{figure}
	\centering
	\begin{subfigure}[b]{0.24\textwidth}
		\centering
		\includegraphics[width=.85\textwidth]{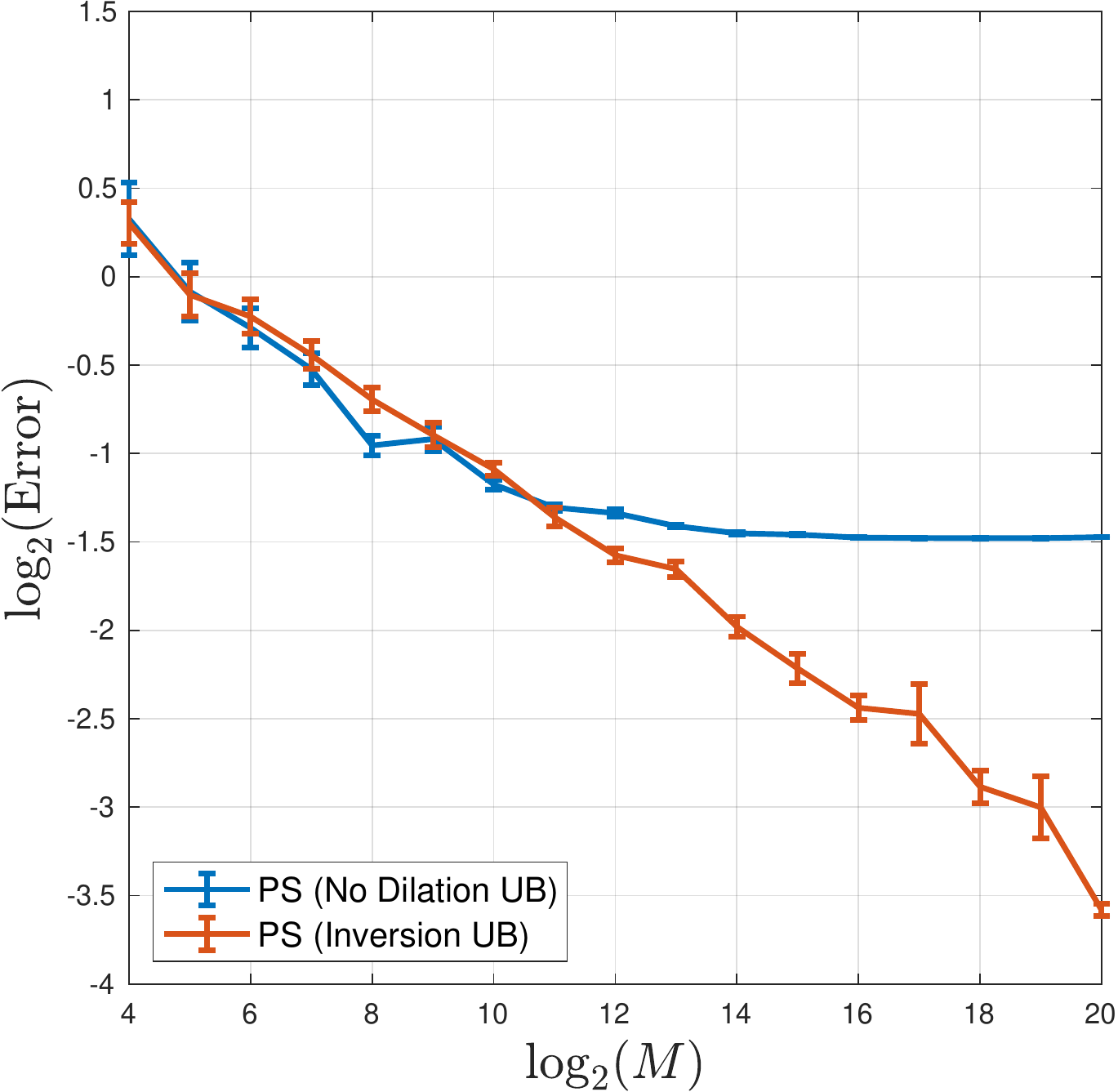}
		\caption{$f_1$ (slope $= -0.2355$)}
		\vspace*{.1cm}
		\label{fig:empirical_f1}
	\end{subfigure}
	\hfill
	\begin{subfigure}[b]{0.24\textwidth}
		\centering
		\includegraphics[width=.85\textwidth]{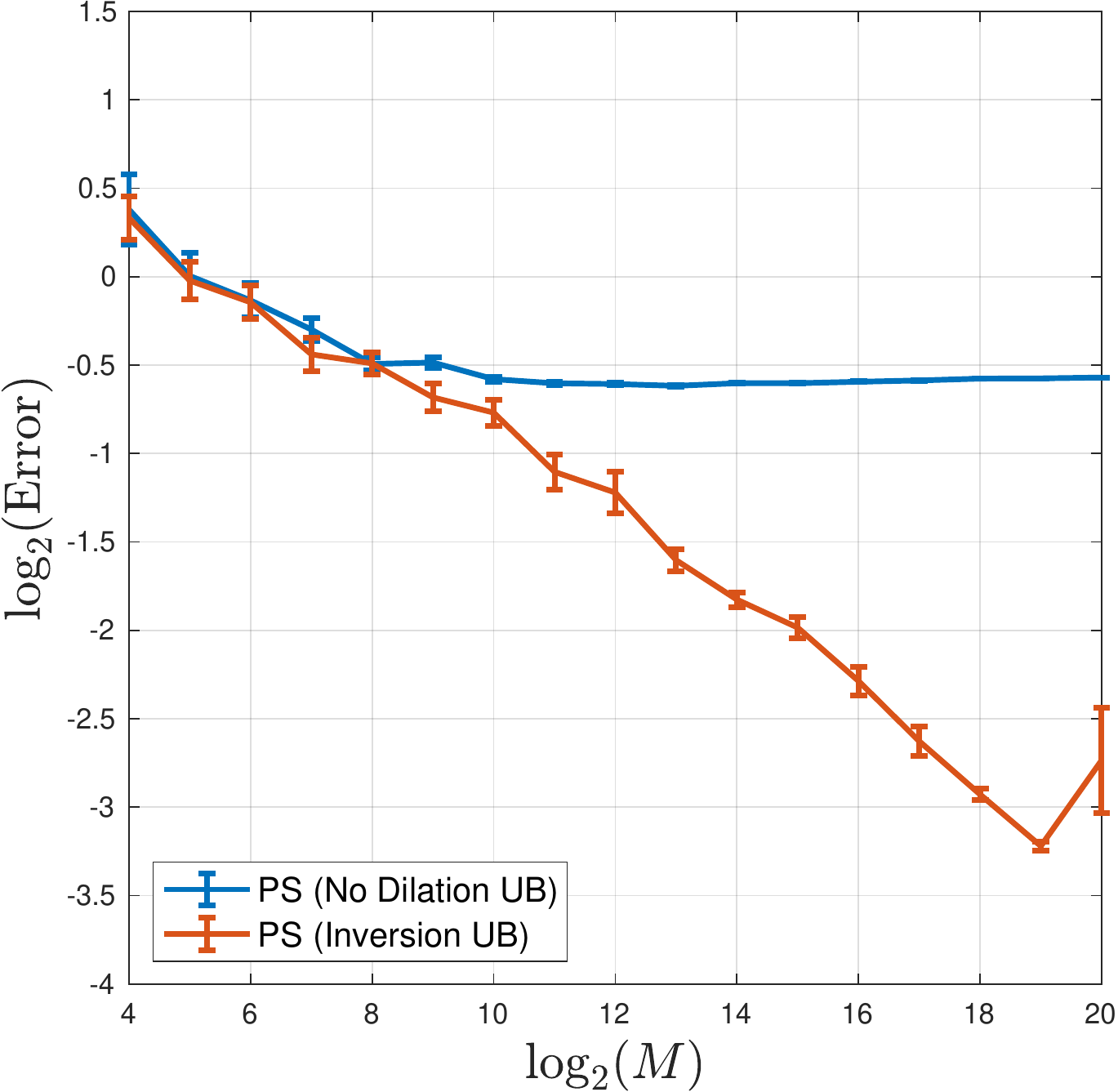}
		\caption{$f_2$ (slope $=-0.2292$)}
		\vspace*{.1cm}
		\label{fig:empirical_f2}
	\end{subfigure}
	\hfill
	\begin{subfigure}[b]{0.24\textwidth}
		\centering
		\includegraphics[width=.85\textwidth]{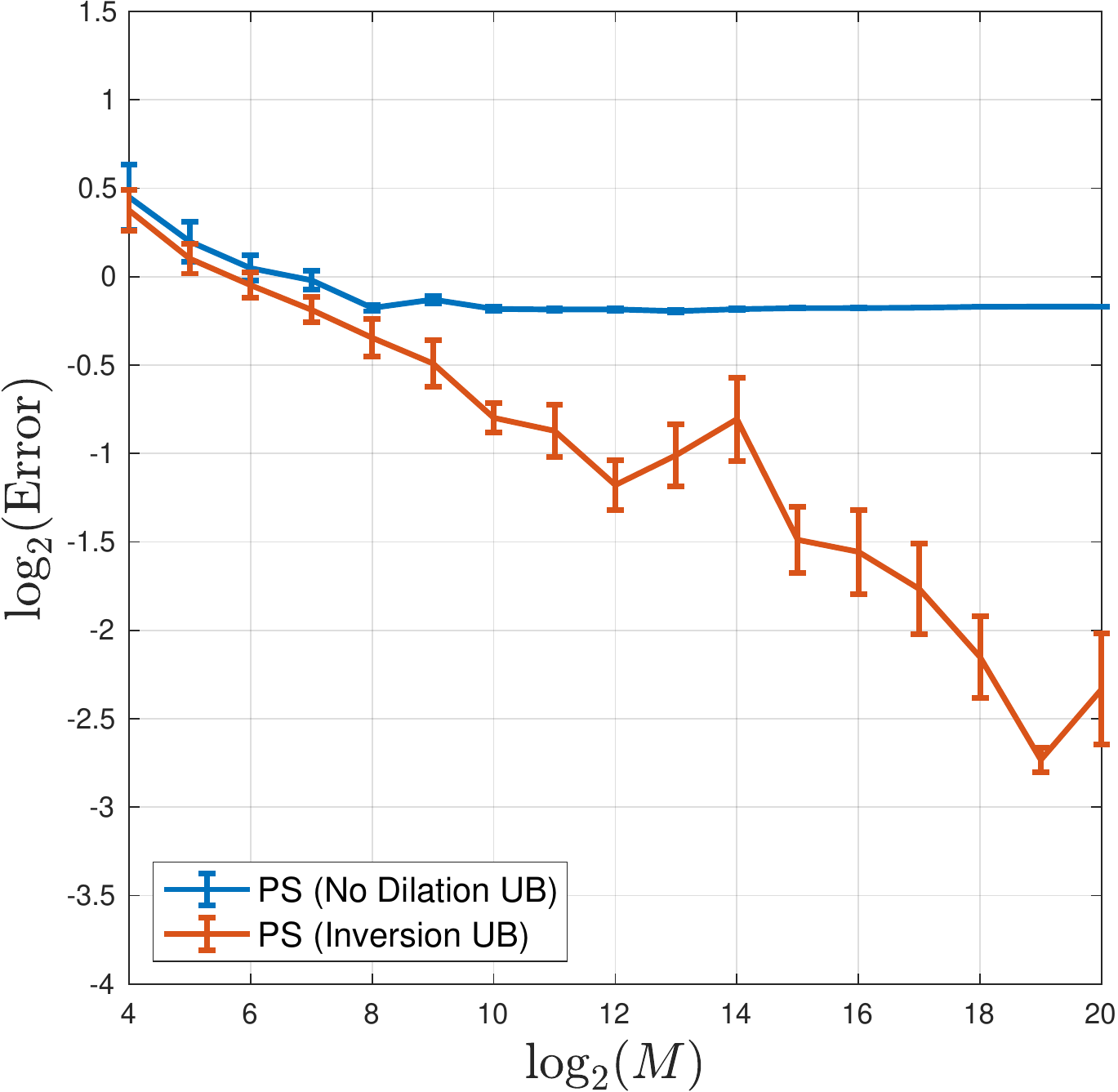}
		\caption{$f_3$ (slope $=-0.2126$)}
		\vspace*{.1cm}
		\label{fig:empirical_f3}
	\end{subfigure}
	\begin{subfigure}[b]{0.24\textwidth}
		\centering
		\includegraphics[width=.85\textwidth]{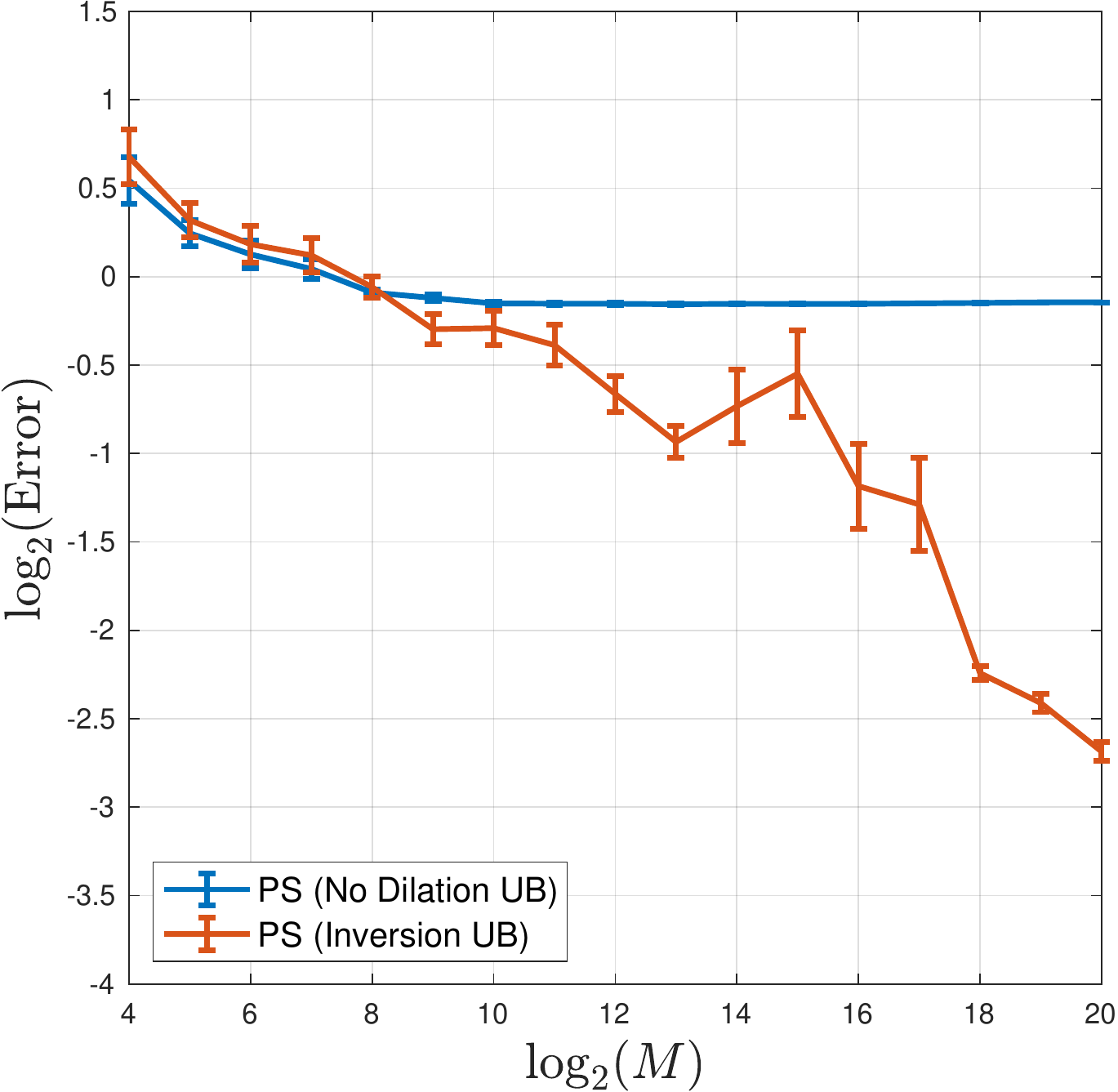}
		\caption{$f_4$ (slope $ = -0.2712$)}
		\vspace*{.1cm}
		\label{fig:empirical_f4}
	\end{subfigure}
	\hfill
	\begin{subfigure}[b]{0.24\textwidth}
		\centering
		\includegraphics[width=.85\textwidth]{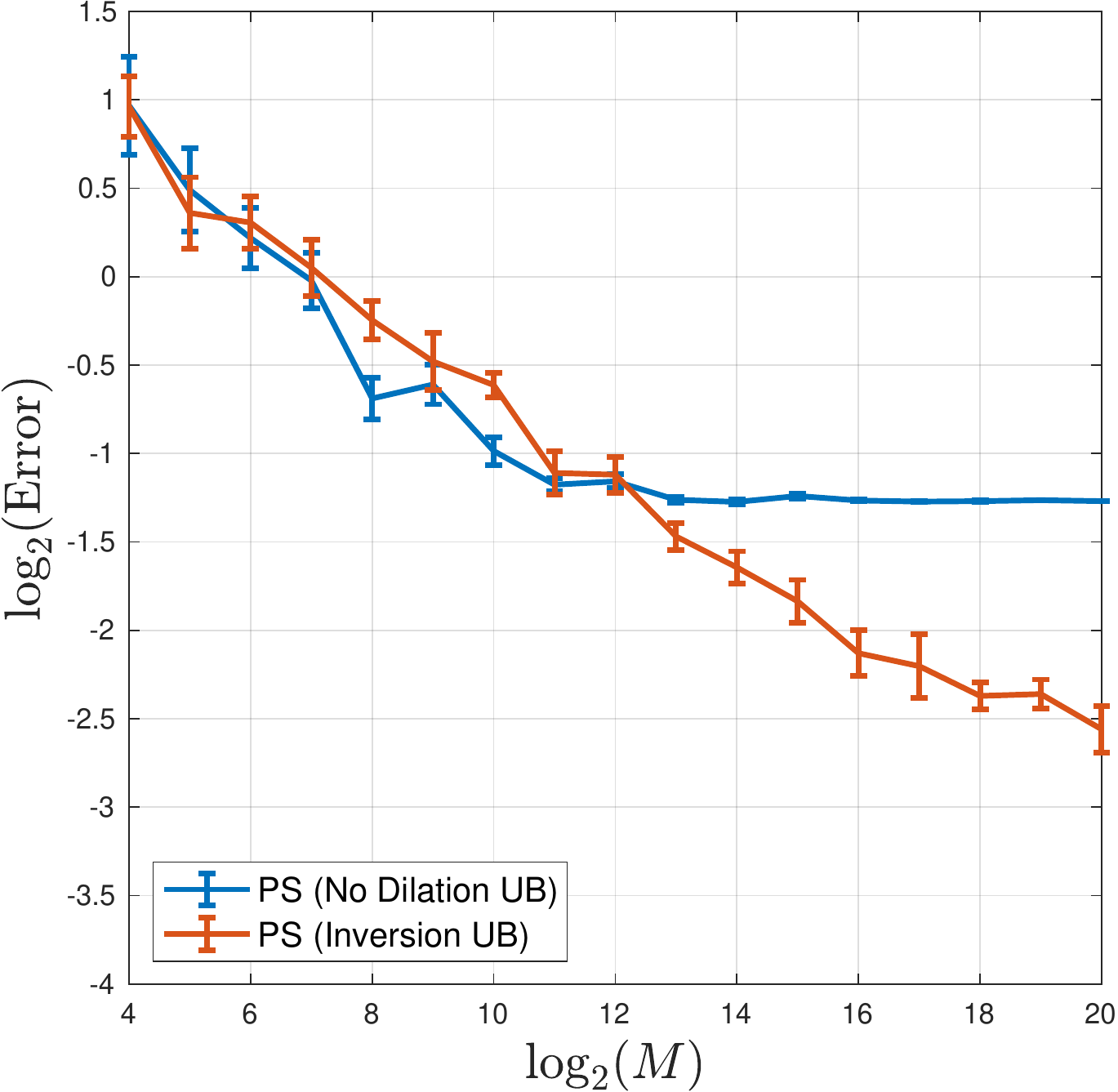}
		\caption{$f_5$ (slope $= -0.1709$)}
		\vspace*{.1cm}
		\label{fig:empirical_f5}
	\end{subfigure}
	\hfill
	\begin{subfigure}[b]{0.24\textwidth}
		\centering
		\includegraphics[width=.85\textwidth]{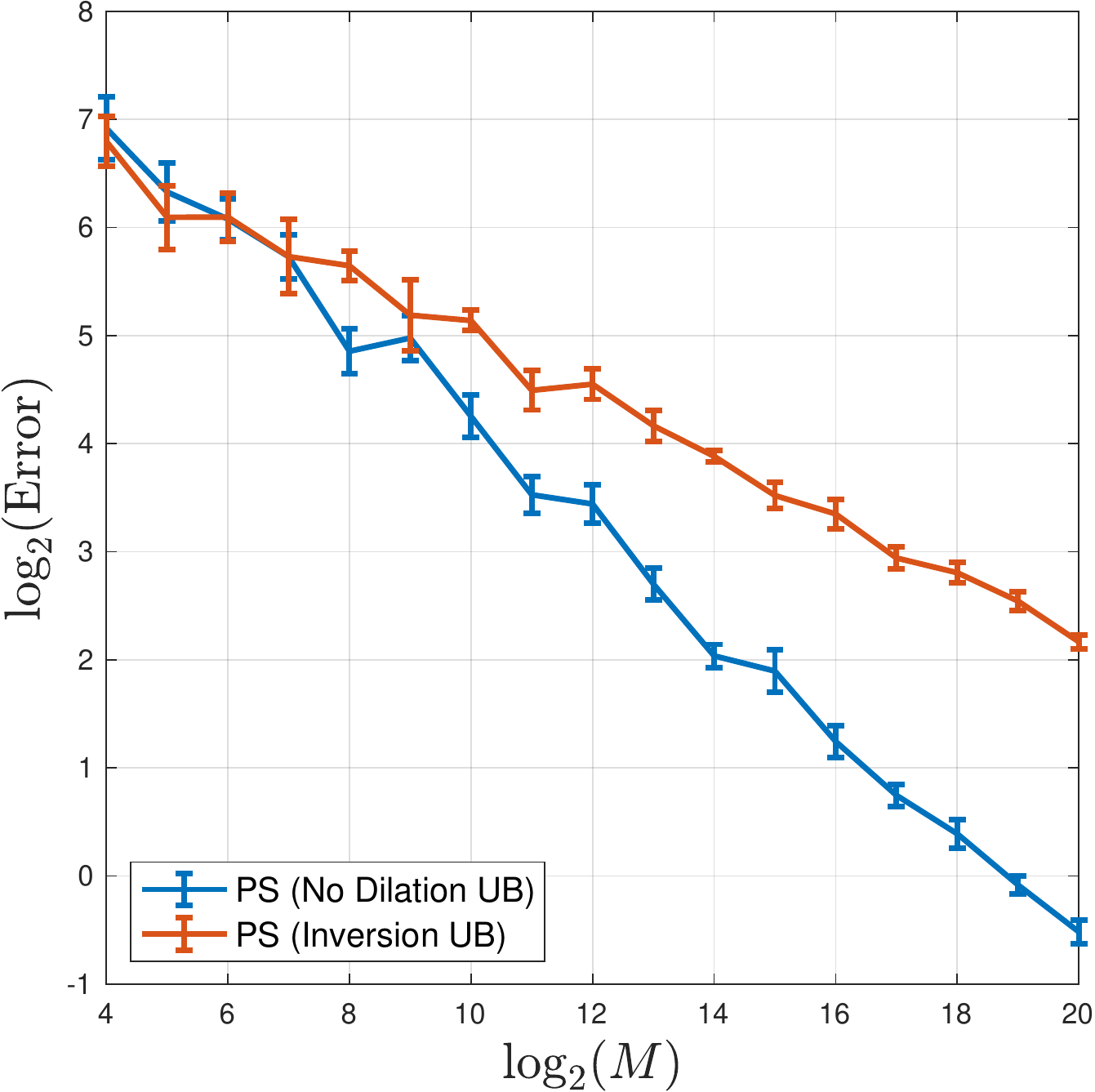}
		\caption{$f_8$ (slope $= -0.2856$)}
		\vspace*{.1cm}
		\label{fig:empirical_f8}
	\end{subfigure}
	\caption{Error decay with standard error bars for Model \ref{model:genMRA} (empirical moment estimation).  All plots show relative $\Lb^2$ error and have the same axis limits, except Figure \ref{fig:empirical_f8}, which shows absolute error. Reported slopes were computed by linear regression on the right half of the plot, i.e. for $12\leq \log_2(M)\leq 20$.}
	\label{fig:GenMRAModelEmpirical}
\end{figure}

More specifically, for signals with a smooth power spectrum ($f_1, \ldots, f_5, f_8$), Corollary \ref{cor:error_for_best_L} predicts that the error should decay like $M^{-1/3}$, i.e. we would expect to observe a slope of $-1/3$ in the log-log plots. In practice the error decay is slightly slower, with a slope of about $-1/4$ for the smooth signals. There are two reasons for the small mismatch between the theory and simulations. First of all, Lemmas \ref{lem:approx_identity_precise} and \ref{lem:approx_identity_hw} are based on Taylor expansions about $L=0$, and so the decay rates in terms of $L$ are only sharp for $L$ small enough; the decay rate is slightly worse in the range of $L$ values used in our simulations; see Figure \ref{fig:Smoothing}. 
In practice when the continuous theory is implemented on a computer, one can never take $L$ smaller than the discrete frequency resolution. 
Secondly, the proof of Theorem \ref{thm:main} applies Young's Convolution Inequality to control the additive noise terms, but simulations indicate that the actual decay rate of the additive noise terms is smaller than this upper bound for the simulation range of $M$ values. See Figure \ref{fig:Young}; as $M \rightarrow \infty$, the decay rates do converge. 

For the non-smooth signals, recall that $f_7$ has a power spectrum which is continuous but not differentiable while $f_6$ has a discontinuous power spectrum. The decay rate of $f_7$ matches that of the smooth signals, but $Pf_7 \notin \mathbf{C}^1 (\R)$, indicating that perhaps $Pf \in \mathbf{C}^3 (\R)$ is not required to achieve the rate in Theorem \ref{thm:main} but an artifact of the proof technique. We note the infinite sample result (Proposition \ref{prop:InfiniteSampleRecoveryPS}) holds under the much milder assumption $Pf \in \mathbf{C}^0 (\R)$. In practice, the decay rate seems to be driven by the $\alpha$ appearing in Lemma \ref{lem:approx_identity_precise}; for $f_6$, Lemma \ref{lem:approx_identity_precise} would apply with $\alpha=1$ to give an error decay like $\sqrt{L}$ and a predicted slope of $-1/12 = -0.083$; we observe $-0.1071$ in Figure \ref{fig:oracle_f6}.

We next investigate the ability of inversion unbiasing to solve Model \ref{model:genMRA} without oracle knowledge of the variances $\sigma^2, \eta^2$. The additive noise level can be reliably estimated from the signal tails. Estimating $\eta$ is more complex and we implement a joint optimization procedure to simultaneously learn $\eta$ and $Pf$. The optimization to learn $\eta$ must be constrained since $\eta=0$ minimizes the loss function (recall Proposition \ref{prop:uniqueCP} only applies to $\eta>0$); $\eta$ is thus constrained to lie in the interval $[0.05, 0.40]$ and we initialize $\eta$ on a course grid ranging from 0.10 to 0.35. For each initialization, the learned $\eta$ value is recorded; a set of candidate $\eta$ values is obtained by discarding learned $\eta$ values which are close to the boundary, and $\eta$ is then selected as the candidate value with the smallest loss. Results are shown in Figure \ref{fig:GenMRAModelEmpirical} for the signals with smooth power spectra; error decay is similar to the oracle case but more variable. Note $\eta$ cannot be reliably learned with this gradient descent procedure when the power spectrum is not smooth.

\section{Conclusion}
\label{sec:conc}

This article considers a generalization of MRA which includes random dilations in addition to random translations and additive noise. The proposed method has several desirable properties compared with previous work. The bias due to dilations is eliminated (not just reduced as in \cite{hirn2020wavelet}) as the sample size increases. In addition, the method is numerically stable, as the unbiasing procedure operates directly on the power spectrum, rather than features derived from the power spectrum. 

There are many compelling directions for future research. By extending the inversion unbiasing procedure to operate on the bispectrum instead of the power spectrum, full signal recovery should be possible with an additional computational cost. In addition, preliminary work suggests that inversion unbiasing can be extended to a broad class of dilation distributions as long as their underlying density functions are known. Thus innovative methods for robustly learning the dilation distribution are critical for these methods to become competitive for real world applications.  Another promising direction is to design a representation which is both translation and dilation invariant, and where the effect of the additive noise can be removed by an averaging procedure. However it remains to be seen whether such a representation exists which is also invertible, i.e. is the hidden signal uniquely defined up to the desired invariants? Once these foundational questions are answered, extensions to 2-dimensional signals are also of interest.


%

\appendices


\section{Proof of Lemma \ref{lem:approx_identity_precise}}
\label{app: approx_identity_precise proof}

\begin{proof}
	Note by assumption there exist constants $C>0$, $\omega_0\geq 1$ such that $|\widehat{h} (\omega)| \leq C|\omega|^{-\alpha}$ for $|\omega| \geq \omega_0$. Also note that $\widehat{\phi}_L(\omega)=e^{-L^2\omega^2/2}$, so that $1-\widehat{\phi}_L(\omega) = \frac{L^2\omega^2}{2}+O(L^3)$ for small $L$. We have:
	\begin{align*}
		\| h - h \ast \phi_L \|_2^2 &= (2\pi)^{-1} \| \widehat{h}(1 - \widehat{\phi}_L )\|_2^2 \\
		&= \frac{1}{2\pi} \int_{|\omega|<\omega_0} |\widehat{h}(\omega)|^2|1-\widehat{\phi}_L(\omega)|^2 \ d\omega \\
		&\quad+ \frac{1}{2\pi} \int_{|\omega|\geq \omega_0} C^2|\omega|^{-2\alpha}|1-\widehat{\phi}_L(\omega)|^2 \ d\omega \\
		&:= (I) + (II) \, .
	\end{align*}
	Note:
	\begin{align*}
		(I) &\leq \int_{|\omega|<\omega_0} |\widehat{h}(\omega)|^2|1-\widehat{\phi}_L(\omega)|^2 \ d\omega \\
		&\leq 2\int_0^{\omega_0} |\widehat{h}(\omega)|^2\left( \frac{L^2\omega^2}{2}+O(L^3)\right)^2\ d\omega \\
		&\leq 2\left(\frac{L^4\omega_0^4}{4}+O(L^5)\right) \int_0^{\omega_0}|\widehat{h}(\omega)|^2 \ d\omega \\
		&\leq \frac{\omega_0^4}{2}\|h\|_2^2L^4+O(L^5) \, .
	\end{align*}	
	To control the second term, note
	\begin{align*}
		(II) &\leq 2C^2 \int_{1}^{\infty} \omega^{-2\alpha}\left(1-e^{-\frac{L^2\omega^2}{2}}\right)^2\ d\omega \\
		&=2C^2  \int_{L}^{\infty} \left(\frac{L}{\tilde{\omega}}\right)^{2\alpha}\left(1-e^{-\frac{\tilde{\omega}^2}{2}}\right)^2\ \frac{d\tilde{\omega}}{L} \\
		&=2C^2L^{2\alpha-1} \int_L^{\infty} \omega^{-2\alpha}\left(1-e^{-\frac{\omega^2}{2}}\right)^2\ d\omega \, .
	\end{align*}
	Explicit evaluation of the upper bound with a computer algebra system gives:
	\begin{align*}
		\alpha=1: &\qquad C_1L + O(L^4) \\
		\alpha=2: &\qquad C_2L^3 + O(L^4) \\
		\alpha=3: &\qquad C_3L^4 + O(L^5)
	\end{align*}
	Also since 
	\begin{align*}
		&\frac{d}{d\alpha}  \int_{1}^{\infty} \omega^{-2\alpha}\left(1-e^{-\frac{L^2\omega^2}{2}}\right)^2\ d\omega \\
		&\quad =  \int_{1}^{\infty} -2\ln(\omega)\omega^{-2\alpha}\left(1-e^{-\frac{L^2\omega^2}{2}}\right)^2\ d\omega <0\,  ,
	\end{align*}
	the upper bound is decreasing in $\alpha$, and we can conclude $(II) \lesssim L^{4\wedge(2\alpha-1)}$ and the lemma is proved. 
\end{proof}

\section{Proof of Lemma \ref{lem:approx_identity_hw}}
\label{app: approx_identity_hw proof}

\begin{proof}
    First observe:
	\begin{align*}
		\| x(h - h \ast \phi_L) \|^2_2 &= (2\pi)^{-1} \norm[\Big]{\frac{d}{d\omega}\left(\widehat{h} - \widehat{h} \widehat{\phi}_L\right)}^2_2 \\
		&= (2\pi)^{-1} \| \widehat{h}' - \widehat{h}' \widehat{\phi}_L-\widehat{h} \widehat{\phi}_L' \|^2_2 \\ &\lesssim \| \widehat{h}' - \widehat{h}' \widehat{\phi}_L\|_2^2 +\|\widehat{h} \widehat{\phi}_L' \|^2_2 \, .
	\end{align*}
	To bound the first term, we apply Lemma \ref{lem:approx_identity_precise} to the function $xh$ to obtain,
	\begin{align*}
		\| \widehat{h}' - \widehat{h}' \widehat{\phi}_L\|_2^2 &= 2\pi \| xh - (xh)\ast \phi_L\|_2^2 \\
		&\lesssim \|xh\|_2^2L^4 + L^{4\wedge(2\alpha-1)} \, .
	\end{align*}
	To bound the second term, note $\widehat{\phi}_L'(\omega) = -L^2\omega e^{-L^2\omega^2/2}$, and that $\|\omega^2 e^{-L^2\omega^2}\|_{\infty} = (eL)^{-1}$. Thus
	\begin{align*}
		\|\widehat{h} \widehat{\phi}_L' \|^2_2 &= L^4 \int |\widehat{h}(\omega)|^2 \omega^2e^{-L^2\omega^2}\ d\omega \leq L^3 \norm{h}_2^2 \, .
	\end{align*}
	Note we could get a higher power for $L$ by
	\begin{align*}
		\|\widehat{h} \widehat{\phi}_L' \|^2_2 &\leq L^4 \int \omega^2|\widehat{h}(\omega)|^2 \ d\omega = L^4 \norm{\omega\widehat{h}}_2^2 \lesssim L^4 \norm{h'}_2^2\, ,
	\end{align*}
	which proves the lemma.
\end{proof}


\section*{Acknowledgment}

This work was supported by the National Science Foundation [grant DMS-1912906 to A.L. and M.H.; grant DMS-1845856 to M.H.], the National Institutes of Health [grant NIGMS-R01GM135929 to M.H.], and the Department of Energy [grant DE-SC0021152 to M.H.].

\ifCLASSOPTIONcaptionsoff
  \newpage
\fi



\bibliographystyle{IEEEtran}
\bibliography{IEEE_MRA}
%

%

\begin{IEEEbiography}[{\includegraphics[width=1in,height=1.25in,clip,keepaspectratio]{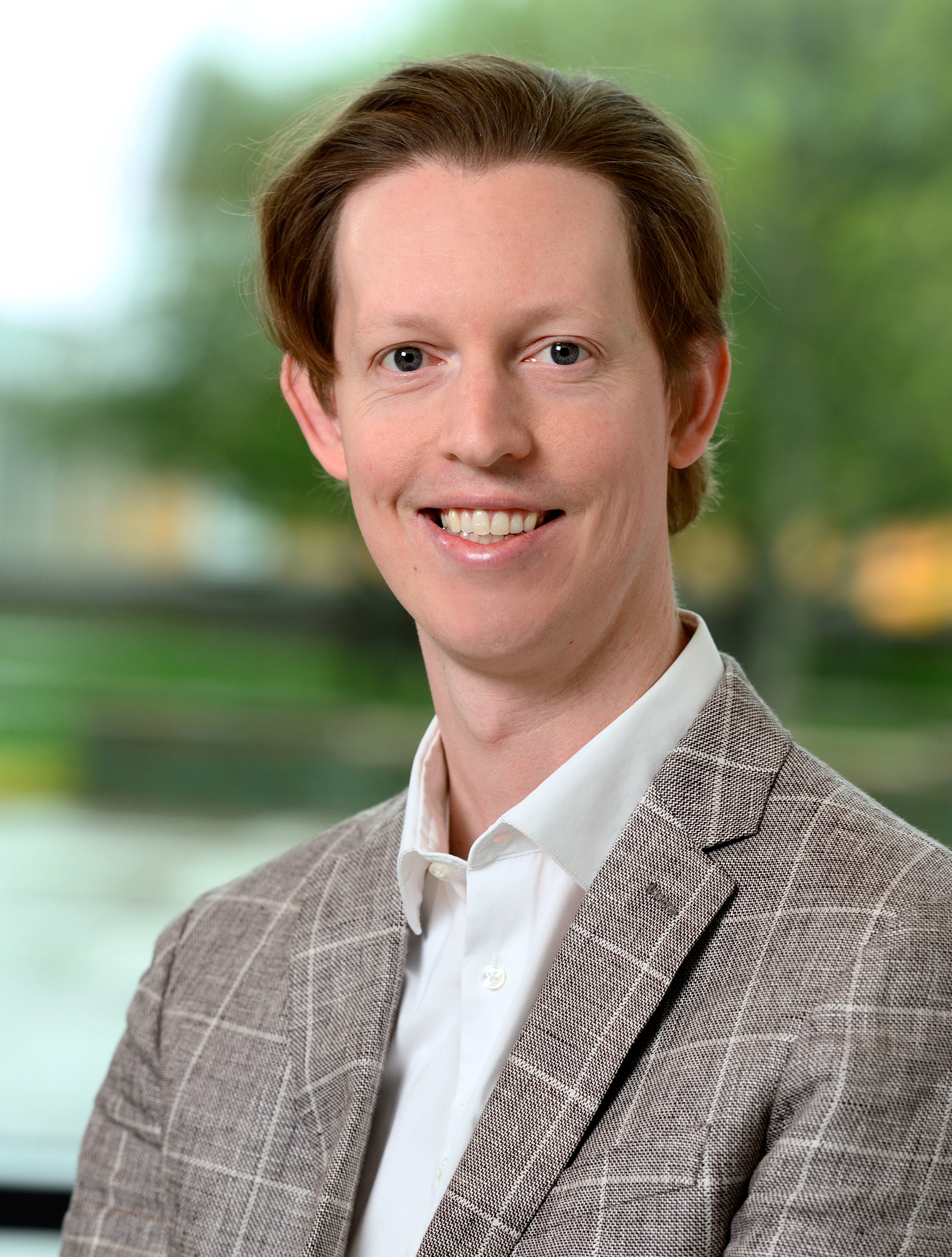}}]{Matthew Hirn}
Matthew Hirn is an Associate Professor in the Department of Computational Mathematics, Science \& Engineering and the Department of Mathematics at Michigan State University. At Michigan State he is the scientific leader of the ComplEx Data Analysis Research (CEDAR) team, which develops new tools in computational harmonic analysis, machine learning, and data science for the analysis of complex, high dimensional data. Hirn received his B.A. in Mathematics from Cornell University and his Ph.D. in Mathematics from the University of Maryland, College Park. Before arriving at MSU, he held postdoctoral appointments in the Applied Math Program at Yale University and in the Department of Computer Science at Ecole Normale Superieure, Paris. He is the recipient of the Alfred P. Sloan Fellowship (2016), the DARPA Young Faculty Award (2016), the DARPA Director’s Fellowship (2018), and the NSF CAREER award (2019), and was designated a Kavli Fellow by the National Academy of Sciences (2017).
\end{IEEEbiography}

\begin{IEEEbiography}[{\includegraphics[width=1in,height=1.25in,clip,keepaspectratio]{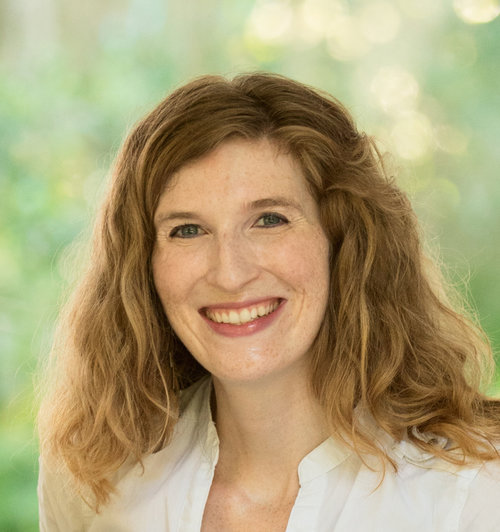}}]{Anna Little}
	Anna Little received her PhD from Duke University in 2011, where she worked under Mauro Maggioni to develop a new multiscale method for estimating the intrinsic dimension of a data set. From 2012-2017 she was an Assistant Professor of Mathematics at Jacksonville University, a primarily undergraduate liberal arts institution where in addition to teaching and research she served as a statistical consultant. From 2018-2020 she was a research postdoc in the Department of Computational Mathematics, Science, and Engineering at Michigan State University, where she worked with Yuying Xie and Matthew Hirn on statistical and geometric analysis of high-dimensional data. She is currently an Assistant Professor in the Department of Mathematics at the University of Utah, as well as a member of the Utah Center for Data Science.
\end{IEEEbiography}

\vfill






\end{document}